\documentclass[11pt]{article} 
\usepackage{algorithmic}
\usepackage{algorithm}
\usepackage{times} 
\usepackage{amsfonts}
\usepackage{amsmath} 
\usepackage{latexsym} 
\usepackage{epsfig}
\usepackage{latexsym} 
\usepackage{verbatim}
\begin{document}
 
\newtheorem{theorem}{Theorem}
\newtheorem{corollary}[theorem]{Corollary}
\newtheorem{prop}[theorem]{Proposition} 
\newtheorem{problem}[theorem]{Problem}
\newtheorem{lemma}[theorem]{Lemma} 
\newtheorem{remark}[theorem]{Remark}
\newtheorem{observation}[theorem]{Observation}
\newtheorem{defin}{Definition} 
\newtheorem{example}{Example}
\newtheorem{conj}{Conjecture} 
\newenvironment{proof}{{\bf Proof:}}{\hfill$\Box$} 
\newcommand{\PR}{\noindent {\bf Proof:\ }} 
\def\EPR{\hfill $\Box$\linebreak\vskip.5mm} 
 
\def\Pol{{\sf Pol}} 
\def\mPol{{\sf MPol}} 
\def\Polo{{\sf Pol}_1} 
\def\PPol{{\sf pPol\;}} 
\def\Inv{{\sf Inv}}
\def\mInv{{\sf MInv}} 
\def\Clo{{\sf Clo}\;} 
\def\Con{{\sf Con}} 
\def\concom{{\sf Concom}\;} 
\def\End{{\sf End}\;}
\def\Sub{{\sf Sub}} 
\def\Im{{\sf Im}} 
\def\Ker{{\sf Ker}\;} 
\def\H{{\sf H}}
\def\S{{\sf S}} 
\def\D{{\sf P}} 
\def\I{{\sf I}} 
\def\Var{{\sf var}} 
\def\PVar{{\sf pvar}} 
\def\fin#1{{#1}_{\rm fin}}
\def\P{{\sf P}} 
\def\Pfin{{\sf P_{\rm fin}}} 
\def\Id{{\sf Id}}
\def\R{{\rm R}} 
\def\F{{\rm F}} 
\def\Term{{\sf Term}}
\def\var#1{{\sf var}(#1)} 
\def\Sg#1{{\sf Sg}(#1)} 
\def\Sgo#1{{\sf Sg}_{\mathrm{old}}(#1)} 
\def\Sgn#1{{\sf Sg}_{\mathrm{new}}(#1)} 
\def\Sgg#1#2{{\sf Sg}_{#1}(#2)} 
\def\Cg#1{{\sf Cg}(#1)}
\def\Cgg#1#2{{\sf Cg}_{#1}(#2)} 
\def\tol{{\sf tol}} 
\def\rbcomp#1{{\sf rbcomp}(#1)}
  
\let\cd=\cdot 
\let\eq=\equiv 
\let\op=\oplus 
\let\omn=\ominus
\let\meet=\wedge 
\let\join=\vee 
\let\tm=\times
\def\ldiv{\mathbin{\backslash}} 
\def\rdiv{\mathbin/}
  
\def\typ{{\sf typ}} 
\def\zz{{\un 0}} 
\def\zo{{\un 1}}
\def\one{{\bf1}} 
\def\two{{\bf2}} 
\def\three{{\bf3}}
\def\four{{\bf4}} 
\def\five{{\bf5}}
\def\pq#1{(\zz_{#1},\mu_{#1})}
  
\let\wh=\widehat 
\def\ox{\ov x} 
\def\oy{\ov y} 
\def\oz{\ov z}
\def\of{\ov f} 
\def\oa{\ov a} 
\def\ob{\ov b} 
\def\oc{\ov c}
\def\od{\ov d} 
\def\oob{\ov{\ov b}} 
\def\rx{{\rm x}}
\def\rf{{\rm f}} 
\def\rrm{{\rm m}} 
\let\un=\underline
\let\ov=\overline 
\let\cc=\circ 
\let\rb=\diamond 
\def\ta{{\tilde a}} 
\def\tz{{\tilde z}}
  
  
\def\zZ{{\mathbb Z}} 
\def\B{{\mathcal B}} 
\def\P{{\mathcal P}}
\def\zL{{\mathbb L}} 
\def\zD{{\mathbb D}}
 \def\zE{{\mathbb E}}
\def\zG{{\mathbb G}} 
\def\zA{{\mathbb A}} 
\def\zB{{\mathbb B}}
\def\zC{{\mathbb C}} 
\def\zM{{\mathbb M}} 
\def\zR{{\mathbb R}}
\def\zS{{\mathbb S}} 
\def\zT{{\mathbb T}} 
\def\zN{{\mathbb N}}
\def\zQ{{\mathbb Q}} 
\def\zW{{\mathbb W}} 
\def\bK{{\bf K}}
\def\C{{\bf C}} 
\def\M{{\bf M}} 
\def\E{{\bf E}} 
\def\N{{\bf N}}
\def\O{{\bf O}} 
\def\bN{{\bf N}} 
\def\bX{{\bf X}} 
\def\GF{{\rm GF}} 
\def\cC{{\mathcal C}} 
\def\cA{{\mathcal A}}
\def\cB{{\mathcal B}} 
\def\cD{{\mathcal D}} 
\def\cE{{\mathcal E}} 
\def\cF{{\mathcal F}} 
\def\cG{{\mathcal G}} 
\def\cH{{\mathcal H}}
\def\cI{{\mathcal I}} 
\def\cL{{\mathcal L}} 
\def\cP{{\mathcal
P}} \def\cR{{\mathcal R}} 
\def\cRY{{\mathcal RY}}
\def\cS{{\mathcal S}} 
\def\cT{{\mathcal T}} 
\def\oB{{\ov B}}
\def\oC{{\ov C}} 
\def\ooB{{\ov{\ov B}}} 
\def\ozB{{\ov{\zB}}}
\def\ozD{{\ov{\zD}}} 
\def\ozG{{\ov{\zG}}}
\def\tcA{{\widetilde\cA}} 
\def\tcC{{\widetilde\cC}}
\def\tcF{{\widetilde\cF}} 
\def\tcI{{\widetilde\cI}}
\def\tB{{\widetilde B}} 
\def\tC{{\widetilde C}}
\def\tD{{\widetilde D}} 
\def\ttB{{\widetilde{\widetilde B}}}
\def\ttC{{\widetilde{\widetilde C}}}
\def\tba{{\tilde\ba}} 
\def\ttba{{\tilde{\tilde\ba}}}
\def\tbb{{\tilde\bb}} 
\def\ttbb{{\tilde{\tilde\bb}}}
\def\tbc{{\tilde\bc}} 
\def\tbd{{\tilde\bd}}
\def\tbe{{\tilde\be}} 
\def\tbt{{\tilde\bt}}
\def\tbu{{\tilde\bu}} 
\def\tbv{{\tilde\bv}}
\def\tbw{{\tilde\bw}} 
\def\tdl{{\tilde\dl}} 
\def\ocP{{\ov\cP}}
\def\tzA{{\widetilde\zA}} 
\def\tzC{{\widetilde\zC}}
\def\new{{\mbox{\footnotesize new}}}
\def\old{{\mbox{\footnotesize old}}}
\def\prev{{\mbox{\footnotesize prev}}}
\def\oo{{\mbox{\sf\footnotesize o}}}
\def\pp{{\mbox{\sf\footnotesize p}}}
\def\nn{{\mbox{\sf\footnotesize n}}} 
\def\oR{{\ov R}}
\def\bA{\mathbf{R}}
  
  
\def\gA{{\mathfrak A}} 
\def\gV{{\mathfrak V}} 
\def\gS{{\mathfrak S}} 
\def\gK{{\mathfrak K}} 
\def\gH{{\mathfrak H}}
  
\def\ba{{\bf a}} 
\def\bb{{\bf b}} 
\def\bc{{\bf c}} 
\def\bd{{\bf d}} 
\def\be{{\bf e}} 
\def\bbf{{\bf f}} 
\def\bg{{\bf g}}
\def\bh{{\bf h}}
\def\bi{{\bf i}} 
\def\bm{{\bf m}} 
\def\bo{{\bf o}} 
\def\bp{{\bf p}} 
\def\bs{{\bf s}} 
\def\bu{{\bf u}} 
\def\bt{{\bf t}} 
\def\bv{{\bf v}} 
\def\bx{{\bf x}}
\def\by{{\bf y}} 
\def\bw{{\bf w}} 
\def\bz{{\bf z}}
\def\ga{{\mathfrak a}} 
\def\oal{{\ov\al}} 
\def\obeta{{\ov\beta}}
\def\ogm{{\ov\gm}} 
\def\oep{{\ov\varepsilon}}
\def\oeta{{\ov\eta}} 
\def\oth{{\ov\th}} 
\def\ovm{{\ov\mu}}
\def\ozero{{\ov0}}
  
  
\def\CCSP{\hbox{\rm c-CSP}} 
\def\CSP{{\rm CSP}} 
\def\NCSP{{\rm \#CSP}} 
\def\mCSP{{\rm MCSP}} 
\def\FP{{\rm FP}} 
\def\PTIME{{\bf PTIME}} 
\def\GS{\hbox{($*$)}} 
\def\ry{\hbox{\rm r+y}}
\def\rb{\hbox{\rm r+b}} 
\def\Gr#1{{\mathrm{Gr}(#1)}}
\def\Grp#1{{\mathrm{Gr'}(#1)}} 
\def\Grpr#1{{\mathrm{Gr''}(#1)}}
\def\Scc#1{{\mathrm{Scc}(#1)}} 
\def\rel{R} 
\def\relo{Q}
\def\rela{S} 
\def\dep{\mathsf{dep}}
\def\Filt{\mathrm{Ft}}
\def\Filts{\mathrm{Fts}} 
\def\Agr{$\mathbb{A}$}
\def\Al{\mathrm{Alg}}
\def\Sig{\mathrm{Sig}}
\def\strat{\mathsf{strat}}
\def\relmax{\mathsf{relmax}}
\def\srelmax{\mathsf{srelmax}}
\def\Meet{\mathsf{Meet}}
\def\amax{\mathsf{amax}}
\def\as{\mathsf{as}}
\def\star{\hbox{$(*)$}}
\def\bmal{{\mathbf m}}
\def\Af{\mathsf{Af}}
\let\sqq=\sqsubseteq

  
\let\sse=\subseteq 
\def\ang#1{\langle #1 \rangle}
\def\angg#1{\left\langle #1 \right\rangle}
\def\dang#1{\ang{\ang{#1}}} 
\def\vc#1#2{#1 _1\zd #1 _{#2}}
\def\tms#1#2{#1 _1\tm\dots\tm #1 _{#2}}
\def\zd{,\ldots,} 
\let\bks=\backslash 
\def\red#1{\vrule height7pt depth3pt width.4pt
\lower3pt\hbox{$\scriptstyle #1$}}
\def\fac#1{/\lower2pt\hbox{$\scriptstyle #1$}}
\def\me{\stackrel{\mu}{\eq}} 
\def\nme{\stackrel{\mu}{\not\eq}}
\def\eqc#1{\stackrel{#1}{\eq}} 
\def\cl#1#2{\arraycolsep0pt
\left(\begin{array}{c} #1\\ #2 \end{array}\right)}
\def\cll#1#2#3{\arraycolsep0pt \left(\begin{array}{c} #1\\ #2\\
#3 \end{array}\right)} 
\def\clll#1#2#3#4{\arraycolsep0pt
\left(\begin{array}{c} #1\\ #2\\ #3\\ #4 \end{array}\right)}
\def\cllll#1#2#3#4#5#6{ \left(\begin{array}{c} #1\\ #2\\ #3\\
#4\\ #5\\ #6 \end{array}\right)} 
\def\pr{{\rm pr}}
\let\upr=\uparrow 
\def\ua#1{\hskip-1.7mm\uparrow^{#1}}
\def\sua#1{\hskip-0.2mm\scriptsize\uparrow^{#1}} 
\def\lcm{{\rm lcm}} 
\def\perm#1#2#3{\left(\begin{array}{ccc} 1&2&3\\ #1&#2&#3
\end{array}\right)} 
\def\w{$\wedge$} 
\let\ex=\exists
\def\NS{{\sc (No-G-Set)}} 
\def\lev{{\sf lev}}
\let\rle=\sqsubseteq 
\def\ryle{\le_{ry}} 
\def\ryprec{\le_{ry}}
\def\os{\mbox{[}} 
\def\zs{\mbox{]}}
\def\link{{\sf link}}
\def\solv{\stackrel{s}{\sim}} 
\def\mal{\mathbf{m}}
\def\precs{\prec_{as}}

  
\def\lb{$\linebreak$}  
  
\def\ar{\hbox{ar}} 
\def\Im{{\sf Im}\;} 
\def\deg{{\sf deg}}
\def\id{{\rm id}}
  
\let\al=\alpha 
\let\gm=\gamma 
\let\dl=\delta 
\let\ve=\varepsilon
\let\ld=\lambda 
\let\om=\omega 
\let\vf=\varphi 
\let\vr=\varrho
\let\th=\theta 
\let\sg=\sigma 
\let\Gm=\Gamma 
\let\Dl=\Delta
  
  
\font\tengoth=eufm10 scaled 1200 
\font\sixgoth=eufm6
\def\goth{\fam12} 
\textfont12=\tengoth 
\scriptfont12=\sixgoth
\scriptscriptfont12=\sixgoth 
\font\tenbur=msbm10
\font\eightbur=msbm8 
\def\bur{\fam13} 
\textfont11=\tenbur
\scriptfont11=\eightbur 
\scriptscriptfont11=\eightbur
\font\twelvebur=msbm10 scaled 1200 
\textfont13=\twelvebur
\scriptfont13=\tenbur 
\scriptscriptfont13=\eightbur
\mathchardef\nat="0B4E 
\mathchardef\eps="0D3F

\title{Graphs of finite algebras, edges, and connectivity}
\author{Andrei A.\ Bulatov\\ 
} 
\date{} 
\maketitle
\begin{abstract}
We refine and advance the study of the local structure of idempotent finite algebras started
in [A.Bulatov, \emph{The Graph of a Relational Structure and Constraint Satisfaction Problems}, LICS, 2004]. We introduce a 
graph-like structure on an arbitrary finite idempotent algebra omitting type \one. We
show that this graph is connected, its edges can be classified into 3 types corresponding to
the local behavior (semilattice, majority, or affine) of certain term operations, and that
the structure of the algebra can be `improved' without introducing type \one\ by choosing
an appropriate reduct of the original algebra. Then we refine this structure demonstrating 
that the edges of the graph of an algebra can be made `thin', that is, there are term 
operations that behave very similar to semilattice, majority, or affine operations on 
2-element subsets of the algebra. Finally, we prove certain connectivity properties of the refined
structures.

This research is motivated by the study of the Constraint Satisfaction Problem, although
the problem itself does not really show up in this paper.
\end{abstract}

\section{Introduction}\label{sec:introduction}

The study of the Constraint Satisfaction Problem (CSP) and especially the Dichotomy 
Conjecture triggered a wave of research in universal algebra, as it turns out that the
algebraic approach to the CSP developed in \cite{Bulatov05:classifying,Jeavons97:closure}
is the most prolific one in this area. These developments have led to a number of strong results
about the CSP, see, e.g., \cite{Barto11:conservative,Barto14:local,Barto12:near,%
Bulatov06:3-element,Bulatov11:conservative,Bulatov14:conservative,Bulatov06:simple,%
Idziak10:few}. However, successful application of the algebraic approach also requires new results 
about the structure of finite algebras. Two ways to describe this structure have been 
proposed. One is based on absorption properties \cite{Barto15:constraint,Barto12:absorbing}
and has led not only to new results on the CSP, but also to significant developments in
universal algebra itself.

In this paper we refine and advance the alternative approach originally developed in 
\cite{Bulatov04:graph,Bulatov11:conjecture,Bulatov08:recent}, which is based on the 
local structure of finite algebras. This 
approach identifies subalgebras or factors of an algebra having `good' term operations,
that is, operations of one of the three types: semilattice, majority, or affine. It then explores 
the graph or hypergraph formed by such subalgebras, and exploits its connectivity properties.
In a nutshell, this method stems from  the early study of the CSP over so called conservative 
algebras \cite{Bulatov11:conservative}, and has led to a much simpler proof of the dichotomy
conjecture for conservative algebras \cite{Bulatov16:conservative} and to a characterization of 
CSPs solvable by consistency algorithms \cite{Bulatov09:bounded}. In spite of these applications
the original methods suffers from a number of drawbacks that make its use difficult. In the
present paper we refine many of the constructions and fix the deficiencies of the original
method. As in \cite{Bulatov04:graph,Bulatov08:recent} an edge is a pair of elements $a,b$
such that there is a factor algebra of the subalgebra generated by $a,b$ that has an
operation which is semilattice, majority, or affine on the blocks containing $a,b$; this
operation determines the type of edge $ab$. In this paper we allow edges to have more 
than one type if there are several factors witnessing different types. The main difference
from the previous results is the introduction of oriented \emph{thin} majority and afiine edges.
An edge $ab$ is said to thin if there is a term operation that is semilattice on $\{a,b\}$, or 
there is a term operation that satisfies the identities of a majority or affine term (say, in
variables $x,y$) on $\{a,b\}$, but only when $x=a$ and $y=b$. Oriented thin edges
allow us to prove a stronger version of the connectivity of the graph related to an algbera. 
This updated approach makes it possible to give a much simpler proof of the 
result of \cite{Bulatov09:bounded} (see also \cite{Barto14:local}), however, this is a subject
of subsequent papers.

\section{Preliminaries}\label{sec:preliminaries}

In terminology and notation we follow the standard texts on universal algebra
\cite{Burris81:universal,Mckenzie87:algebras}. We also assume familiarity with
the basics of the tame congruence theory \cite{Hobby88:structure}. All algebras in this paper
are assumed to be finite, idempotent, and omitting type \one.

Algebras will be denoted by $\zA,\zB$, etc. The subalgebra of an algebra $\zA$ generated by 
a set $B\sse\zA$ is denoted $\Sgg\zA B$, or if $\zA$ is clear from the context simply by
$\Sg B$. The set of term operations of algebra $\zA$ is denoted by $\Term(\zA)$.
Subalgebras of direct products are often considered as relations. An element
(a tuple) of $\tms\zA n$ is denoted in boldface, say, $\ba$, and its $i$th component
is referred to as $\ba[i]$, that is, $\ba=(\ba[1]\zd\ba[n])$. The set $\{1\zd n\}$ will be denoted 
by $[n]$. For $I\sse[n]$, say, $I=\{\vc ik\}$, $i_1<\dots<i_k$, by $\pr_I\ba$ we denote
the $k$-tuple $(\ba[i_1]\zd\ba[i_k])$, and for $\rel\sse\tms\zA n$ by $\pr_I\rel$
we denote the set $\{\pr_I\ba\mid\ba\in\rel\}$. If $I=\{i\}$ or $I=\{i,j\}$ we write 
$\pr_i,\pr_{ij}$ rather than $\pr_I$. The tuple $\pr_I\ba$ and relation $\pr_I\rel$ are called
the \emph{projections} of $\ba$ and $\rel$ on $I$. A subalgebra (a relation) $\rel$
of $\tms\zA n$ is said to be a \emph{subdirect product} of $\vc\zA n$ if $\pr_i\rel=\zA_i$
for every $i\in[n]$. For a congruence $\al$ of $\zA$ and $a\in\zA$, by $a^\al$ we denote 
the $\al$-block containing $a$, and by $\zA\fac\al$ the factor algebra modulo $\al$.
For $B\sse\zA^2$, the congruence generated by $B$ will be denoted by $\Cgg\zA B$
or just $\Cg B$. By $\zz_\zA,\zo_\zA$ we denote the least (i.e.\ the equality relation),
and the greatest (i.e.\ the total relation) congruence of $\zA$, respectively. Again, we 
often simplify this notation to $\zz,\zo$.

\section{Graph: Thick edges}\label{sec:thick}

\subsection{The three types of edges}\label{sec:three-types}

Let $\zA$ be an algebra with universe $A$.
We introduce graph $\cG(\zA)$  as follows. The vertex set
is the set $A$. A pair $ab$ of vertices is a \emph{edge} if and only if
there exists a congruence $\th$ of $\Sg{a,b}$ and a term operation of $\zA$ such 
that either $f\fac\th$ is an affine 
operation on $\Sg{a,b}\fac\th$, or $f\fac\th$ is a semilattice operation on
$\{a^\th,b^\th\}$, or $f\fac\th$ is a majority operation on
$\{a^\th,b^\th\}$. 

If there exists a congruence and a term operation of $\zA$ such that $f\fac\th$ 
is a semilattice operation on $\{a^\th,b^\th\}$ then $ab$ is said to have the
{\em semilattice type}. An edge $ab$ is of {\em majority type} if there are 
a congruence $\th$ and $f\in\Term(\zA)$ (a term operation of $\zA$, respectively)
 such that $f\fac\th$ is a majority operation on $\{a^\th,b^\th\}$. Finally, $ab$ 
has the {\em affine type} if there are a congruence $\th$ and $f\in\Term(\zA)$ 
(a term operation of $\zA$, respectively) such that $f\fac\th$ is an affine operation 
on $\ang{a,b}\fac\th$.  In all cases we say that congruence $\th$ \emph{witnesses}
the type of edge $ab$.

Note that, for every edge $ab$ of $\cG(\zA)$, there is the associated pair 
$a^\th,b^\th$ from the factor structure. We will need both of these types of pairs 
and will sometimes call $a^\th,b^\th$ a {\em thick} edge (see Fig.~1). The 
smallest congruence certifying the type of an edge $ab$ will be denoted by $\th_{ab}$. 
\begin{figure}[t]
\centerline{\includegraphics{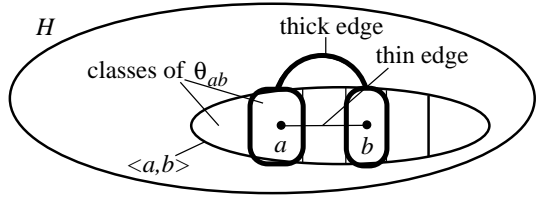}}
\caption{Edges and thick edges}
\end{figure}

Note also that a pair $ab$ may have more than one type witnessed by different 
congruences $\th$. Sometimes we need a stricter version of type. A pair $ab$
is \emph{strictly semilattice} if it is semilattice; $ab$ is said to be \emph{strictly
majority}, if it is majority but not semilattice. Finally, pair $ab$ is said to be 
\emph{strictly affine} if it is affine, but not semilattice or majority.

\subsection{General connectivity}\label{sec:general-connectivity}

\begin{theorem}\label{the:connectedness}
If an idempotent algebra $\zA$ omits type \one, then $\cG(\zB)$ is connected
for every subalgebra of $\zA$.
\end{theorem}

Let $\zA=(A;F)$ be an idempotent algebra. Recall that a {\em
tolerance} of $\zA$ is a binary reflexive and symmetric
relation compatible with $\zA$. The transitive closure of a tolerance is a
congruence of $\zA$. In particular, if $\zA$ is simple then the transitive
closure of every its tolerance different from the equality relation is
the total relation. If a tolerance satisfies this condition then we
say that it is {\em connected}. Let $\tau$ be a tolerance. A set
$B\sse A$ maximal with respect of inclusion and such that $B^2\sse
\tau$ is said to be a {\em class} of $\tau$. We will need the following 
simple observation.
\begin{lemma}\label{lem:tol-class}
Every class of a tolerance of an idempotent algebra is a subalgebra.
\end{lemma}
%

Let $G=(V,E)$ be a hypergraph. A {\em path} in $G$
is sequence $\vc Hk$ of hyperedges such that $H_i\cap H_{i+1}\ne\eps$, 
for $1\le i< k$. The hypergraph $G$ is said to be {\em connected} if,
for any $a,b\in V$, there is a path $\vc Hk$ such that $a\in H_1$,
$b\in H_k$.

Clearly, the universe of an algebra $\zA$ along with the family of all
its proper subalgebras forms a hypergraph denoted by
$\cH(\zA)$. Lemma~\ref{lem:tol-class} implies that, for a simple
idempotent algebra $\zA$, the hypergraph $\cH(\zA)$ is connected unless 
$\zA$ is tolerance free. In the latter case it can be disconnected.

If $\al$ is a congruence of a finite algebra $\zA$ and $\rel$ is a
compatible binary relation, then the {\em $\al$-closure} of $\rel$ is
defined to be $\al\circ\rel\circ\al$. A relation equal to its
$\al$-closure is said to be {\em $\al$-closed}. If $(\al,\beta)$ is a prime
quotient of $\zA$, then the {\em basic tolerance for} $(\al,\beta)$ (see
\cite{Hobby88:structure}, Chapter~5) is the $\al$-closure of the
relation $\al\cup\bigcup\{N^2\mid N$ is an $(\al,\beta)$-trace$\}$ 
if $\typ(\al,\beta)\in\{\two,\three\}$, and it is the
$\al$-closure of the compatible relation generated by
$\al\cup\bigcup\{N^2\mid N$ is an $(\al,\beta)$-trace$\}$ if
$\typ(\al,\beta)\in\{\four,\five\}$. The basic tolerance is the
smallest $\al$-closed tolerance $\tau$ of $\zA$ such that
$\al\ne\tau\sse\beta$. 

Let $(\al,\beta)$ is a prime quotient of $\zA$. An {\em
$(\al,\beta)$-quasi-order} is a compatible reflexive and
transitive relation $\rel$ such that $\rel\cap\rel^{-1}=\al$, and the
transitive closure of $\rel\cup\rel^{-1}$ is $\beta$. The quotient
$(\al,\beta)$ is said to be {\em orderable} if there exists an
$(\al,\beta)$-quasi-order. By Theorem~5.26 of
\cite{Hobby88:structure}, $(\al,\beta)$ is 
orderable if and only if $\typ(\al,\beta)\in\{\four,\five\}$.

Recall that an element $a$ of an algebra $\zA$
is said to be {\em absorbing} if whenever $t(x,\vc yn)$ is an
$(n+1)$-ary term operation of $\zA$ such that $t$ depends on $x$ and
$(\vc bn)\in A^n$, then $t(a,\vc bn)=a$. A congruence $\th$ of 
$\zA^2$ is said to be {\em skew} if it is the kernel of no projection
mapping of $\zA^2$ onto its factors. If $\zA$ is a simple idempotent
algebra, then the result of \cite{Kearnes96:idempotent} states that
one of the following holds: (a) $\zA$ is term equivalent to a module;
(b) $\zA$ has an absorbing element; or (c) $\zA^2$ has no skew
congruence. 

We also need the following easy observation.
\begin{lemma}\label{lem:binary-tol}
Let $\rel$ be an $n$-ary compatible relation on $\zA$ such that, for
any $i\in\{1\zd n\}$, $\pr_i\rel=\zA$. Then, for any $i\in[n]$,
the relation $\tol_i=\{(a,b)\mid $ there are $a_1\zd a_{i-1},a_{i+1}\zd
a_n\in\zA$ such that $(a_1\zd a_{i-1},a,a_{i+1}\zd a_n),\lb(a_1\zd
a_{i-1},b,a_{i+1}\zd a_n)\in\rel\}$ is a tolerance of $\zA$.
\end{lemma}

Tolerance of the form $\tol_i$ will be called \emph{link tolerance}, or
\emph{$i$th link tolerance}
%
\begin{prop}\label{pro:simple}
Let $\zA$ be a simple idempotent algebra.\\[2mm]
(1) If $\typ(\zA)\in\{\four,\five\}$ then $\cH(\zA)$ is connected.\\[2mm]
(2) If $\typ(\zA)=\three$ and $\zA$ has a proper tolerance,
then $\cH(\zA)$ is connected.\\[2mm]
(3) If $\typ(\zA)=\two$ then $\zA$ is term equivalent to a module.\\[2mm]
(4) If $\zA=\Sg{a,b}$, $\typ(\zA)=\three$ and $\zA$ is tolerance free,
then either $a,b$ are connected in $\cH(\zA)$, or there is a binary
term operation $f$ or a ternary term operation $g$ such that $f$ 
is a semilattice operation on $\{a,b\}$, or $g$ is a majority operation
on $\{a,b\}$.
\end{prop}

\begin{proof}
(1) By Theorem~5.26 of \cite{Hobby88:structure}, there exists
$(\zz,\zo)$-quasi-order $\le$ on $\zA$, which is, clearly, just a
compatible partial order. Let $a\le b\in A$ be such that $a\le c\le
b$ implies $c=a$ or $c=b$. We claim that $\{a,b\}$ is a subalgebra of
$\zA$. Indeed, for any term operation $f(\vc xn)$ of $\zA$ and any
$\vc an\in\{a,b\}$, we have $a=f(a\zd a)\le f(\vc an)\le f(b\zd
b)=b$. Finally, it follows from Lemma~5.24(3) and Theorem~5.26(2) that
$\le$ is connected.\\[2mm]
(2) Follows straightforwardly from Lemma~\ref{lem:tol-class} and the
fact that the transitive closure of any proper tolerance of $\zA$ is
the total relation.\\[2mm]
(3) Follows from the results of \cite{Kearnes96:idempotent}.\\[2mm]
(4) We consider two cases.
\medskip

\noindent
{\sc Case 1.} There is no automorphism $\vf$ of $\zA$ such that
$\vf(a)=b$ and $\vf(b)=a$.\\[2mm]
Consider the relation $\rel$ generated by $(a,b),(b,a)$. By the
assumption made, $\rel$ is not the graph of a bijective mapping. By
Lemma~\ref{lem:binary-tol}, $\tol_1,\tol_2$ are tolerances of $\zA$
different from the equality relation. Thus, they are the total
relation. Therefore, there is $c\in\zA$ such that
$(a,c),(b,c)\in\rel$. If both $\Sg{a,c},\Sg{b,c}$ are proper
subalgebras of $\zA$, then $a,b$ are connected in
$\cH(\zA)$. Otherwise, let, say, $\Sg{a,c}=\zA$. Since
$(b,a),(b,c)\in\rel$ and $\zA$ is idempotent, $(b,d)\in\rel$ for any
$d\in \zA$. In particular, $(b,b)\in\rel$. This means that there is a
binary term operation $f$ such that $f(a,b)=f(b,a)=b$, as required.
\medskip

\noindent
{\sc Case 2.} There is an automorphism $\vf$ of $\zA$ such that
$\vf(a)=b$ and $\vf(b)=a$.\\[2mm]
Consider the ternary relation $\rel$ generated by
$(a,a,b),(a,b,a),(b,a,a)$. As in the previous case, if we show that
$(a,a,a)\in\rel$, then the result follows. Let also
$\rela=\{(c,\vf(c))\mid c\in \zA\}$ denotes the graph of an automorphism
$\vf$ with $\vf(a)=b$ and $\vf(b)=a$. 
\medskip

\noindent
{\sc Claim 1.} $\pr_{1,2}\rel=\zA\tm\zA$.\\[2mm]
Let $\relo=\pr_{1,2}\rel$ and $\relo'=\{(c,\vf(d))\mid
(c,d)\in\relo\}$. Since $\relo'(x,y)=\ex z (\relo(x,z)\meet\rela(z,y))$, this
relation is compatible. Clearly, $\relo=\zA\tm\zA$ if and only if
$\relo'=\zA\tm\zA$. Notice that $(a,a),(b,b),(a,b)\in\relo'$. Since
$\typ(\zA)=\three$ and $\zA$ is tolerance free, every pair $c,d\in\zA$
is a trace. Therefore, there is a polynomial operation $g(x)$ with
$g(a)=c, g(b)=d$ and, hence, there is a term operation $f(x,y,z)$ such
that $f(a,b,x)=g(x)$. For this operation we have
$$
f\left(\cl aa,\cl bb,\cl ab\right)=\cl cd\in\relo'.
$$
\medskip

Next we show that $\tol_3$ cannot be the equality relation. Suppose
for contradiction that it is. Then the relation
$\th=\{((c_1,d_1),(c_2,d_2))\mid $ there is $e\in\zA$ such that
$(c_1,d_1,e),(c_2,d_2,e)\in\rel\}$ is a congruence of $\relo=\zA^2$. It
cannot be a skew congruence, hence, it is kernel of the projection of
$\zA^2$ onto one of its factors. Without loss of generality let
$\th=\{((c_1,d_1),(c_2,d_2))\mid c_1=c_2\}$. This means that, for any
$e\in\zA$ and any $(c_1,d_1,e),(c_2,d_2,e)\in\rel$, we have
$c_1=c_2$. However, $(a,b,a),(b,a,a)\in\rel$, a contradiction. The same
argument applies when $\th=\{((c_1,d_1),(c_2,d_2))\mid d_1=d_2\}$.

Thus, $\tol_3$ is the total relation, and there is $(c,d)\in\relo$ such
that $(c,d,a),(c,d,b)\in\rel$ which implies $\{(c,d)\}\tm\zA\sse\rel$. 
\medskip

\noindent
{\sc Claim 2.} For any $(c',d')\in\pr_{1,2}\rel$, the tuple
$(c',d',a)\in\rel$.\\[2mm]
Take a term operation $g(x,y,z)$ such that $g(a,b,c)=c'$ and
$g(a,b,\vf^{-1}(d))=\vf^{-1}(d')$. Such an operation exists whenever 
$c\ne\vf^{-1}(d)$, because every pair of elements of $\zA$ is a trace.
Then
\begin{eqnarray*}
\lefteqn{g\left(\cll aba,\cll baa,\cll cda\right)=\cll{g(a,b,c)}%
{g(b,a,d)}a}\\
&=&\cll{c'}{\vf(g(\vf^{-1}(b),\vf^{-1}(a),\vf^{-1}(d)}a\\
&=&\cll{c'}{\vf(g(a,b,\vf^{-1}(d)}a=
\cll{c'}{\vf(\vf^{-1}(d'))}a=\cll{c'}{d'}a\in\rel. 
\end{eqnarray*}
What is left is to show that there are $c,d$ such that $(c,d,a)\in\rel$ and  $c\ne\vf^{-1}(d)$. 
Suppose $c=\vf^{-1}(d)$. If $\Sg{a,c},\Sg{c,b}\ne\zA$, the $a,b$ are connected in
$\cH(\zA)$.  Let $\Sg{a,c}=\zA$, and $h$ such that $h(a,c)=b$. Since $\rel$ is 
symmetric with respect to any permutation of coordinates, 
$\{c\}\tm\zA\tm\{d\}\sse\rel$. In particular, $(c,c,d)\in\rel$. Then
$$
h\left(\cll aab,\cll ccd\right)=\cll bba,
$$
as $c=\vf^{-1}(d)$ and $a=\vf^{-1}(b)$. The tuple $(b,b,a)$ is as required.

Thus, $(a,a,a)\in\rel$ which means that there is a term operation
$f(x,y,z)$ such that $f(a,a,b)=f(a,b,a)=f(b,a,a)=a$. Since $\vf$ is an 
automorphism, we also get $f(b,b,a)=f(b,a,b)=f(a,b,b)=b$, i.e.\ $f$ is 
a majority operation on $\{a,b\}$.
\end{proof}

\begin{proof}[Theorem~\ref{the:connectedness}]
Suppose for contradiction that $\cG(\zA)$ is disconnected. Let
$\zB$ be a minimal subalgebra of $\zA$ such that $\cG(\zB)$ is
disconnected. Since the graph of every proper subalgebra of $\zB$ is
connected, $\zB$ is 2-generated, say, $\zB=\Sg{a,b}$. 
Let $\th$ be a maximal congruence of $\zB$. 

Clearly, if $\cG(\zB\fac\th)$ is connected then $\cG(\zB)$ is
connected. Therefore, $\zB\fac\th$ is tolerance free and of type \three. Take
$c,d\in\zB$; let $c'=c^\th, d'=d^\th$. If $\Sg{c,d}\ne\zB$ then
$c,d$ are connected by the assumption made. Otherwise
$\Sg{c',d'}=\zB\fac\th$. By Proposition~\ref{pro:simple}, either
$c',d'$ are connected in $\cH(\zB\fac\th)$ and hence in $\cG(\zB\fac\th)$,
or $c'd'$ is an edge in $\cG(\zB\fac\th)$. In the former case $c,d$
are connected because every proper subalgebra of 
$\zB\fac\th$ gives rise to a proper subalgebra of $\zB$. In the latter
case $cd$ is an edge of $\cG(\zB)$. Thus, $\cG(\zB)$ is connected, a
contradiction. 
\end{proof}

\subsection{Adding thick edges}\label{sec:adding-thick}

Generally, an edge, or even a thick edge is not a subalgebra. However, we show 
that every idempotent algebra $\zA$ omitting type \one\ has a reduct $\zA'$ such that
$\zA'$ also omits type one, but every its edge of semilattice or majority type is a 
subalgebra of $\zA'$. Moreover, some type restrictions are also observed.
We say that $\cG(\zA)$ is {\em semilattice} ({\em semilattice/majority})-{\em connected} 
if every two vertices in $\cG(\zA)$ are connected by a path consisting of semilattice
(semilattice and strict majority) edges. For short we will abbreviate it to s-connected and 
sm-connected.
\begin{theorem}\label{the:adding}
Let $\zA$ be an idempotent algebra omitting type \one, 
$ab$ an edge of $\cG(\zA)$ of semilattice or strict majority type, and
$\rel_{ab}=(a^{\th_{ab}}\cup b^{\th_{ab}})$ the thick edge $ab$.
Let also $F_{ab}$ denote set of term operations of $\zA$ preserving $\rel_{ab}$\\[1mm]
(1) $\zA'=(A,F_{ab})$ omits type \one.\\[1mm] 
(2) If $ab$ is semilattice and $\cG(\zA)$ is s-connected, then
$\cG(\zA')$ is s-connected.\\[1mm]
(3) If $ab$ is strict majority and $\cG(\zA)$ is sm-connected, then
$\cG(\zA')$ is sm-connected.
\end{theorem}

We prove Theorem~\ref{the:adding} by induction on the `structure' of the algebra.
The base case of this induction is given by strictly simple algebras.
Recall that a simple algebra whose proper subalgebras are all
1-element is said to be {\em strictly simple}. We need the description
of finite idempotent strictly simple algebras given in
\cite{Szendrei90:surjective}. 

Let $G$ be a permutation group acting on a set $A$. By $\R(G)$ we
denote the set of operations on $A$ preserving each
relation of the form $\{(a,g(a))\mid a\in A\}$ where
$g\in G$, and $\F(G)$ denotes the set of idempotent
members of $\R(G)$.

Let ${}_K\ov A=(A;+,K)$ be a finite dimensional
vector space over a finite field $K$, $T(\ov A)$  the group of
translations 
$\{x+a\mid a\in A\}$, and $\End_K\ov A$ the endomorphism ring
of ${}_K\ov A$. Then one can consider $\ov A$ as a module over
$\End_K\ov A$. This module is denoted by ${}_{(\End_K\ov A)}\ov A$.

Finally, let $\F^0_k$ denote the set of all
operations preserving the relation
\begin{eqnarray*}
&& X^0_k=\{(a_1\zd a_k)\in A^k\mid a_i=0\hbox{ for at least one
$i$,}\\
&&\qquad 1\le i\le k\}
\end{eqnarray*}
where $0$ is some fixed element of $A$,
and let $\F^0_\omega=\bigcap_{k=2}^\infty{\F^0_k}$.
\begin{theorem}[\cite{Szendrei90:surjective}] \label{str}
A finite strictly simple idempotent algebra $\zA$ is 
term equivalent to one of the following algebras:

(a) $(A,\F(G))$ for a permutation group $G$ on $A$
such that every nonidentity member of $G$ has at most one
fixed point;

(b) $(A,\Term_\id({}_{(\End_K\ov A)}\ov A))$ for
some vector space ${}_K\ov A$ over a finite field $K$;

(c) $(A,\F(G)\cap\F^0_k)$ for some $k$ ($2\le k\le\omega$),
some element $0\in A$ and some permutation group $G$ acting on $A$
such that $0$ is the unique fixed point of every nonidentity
member of~$G$;

(d) $(A,F)$ where $|A|=2$ and $F$ contains a semilattice operation;

(e) a two-element algebra with an empty set of basic operations.
\end{theorem}
It can be easily shown (see e.g.\ \cite{Bulatov05:classifying}) that
in case (c) $\zA$ has a term {\em zero-multiplication} operation, that
a binary operation $h$ such that $h(x,y)=0$ whenever $x\ne y$.

\begin{proof}[of Theorem~\ref{the:adding}.]
Let $ab$ be an edge of semilattice type and $f$ is a term operation
such that $f\fac{\th_{ab}}$ is a semilattice operation on
$B'=\{a^{\th_{ab}},b^{\th_{ab}}\}$. We will omit index
$\th_{ab}$ everywhere it does not lead to a confusion. Let
$\zA'=(A;F')$ where $F'$ is the set of binary term operations $g$ of
$\zA$ such that $g\fac{\th_{ab}}$ on $B'$ is either a projection or
equals $f\fac{\th_{ab}}$. The subalgebra of $\zA$ generated by a set
$B\sse A$ will be denoted by $\Sgo B$, while the subalgebra of
$\zA'$ generated by the same set will be denote by $\Sgn B$. In
general, $\Sgn B\sse\Sgo B$.
\medskip

\noindent
{\em Claim 1.} $f$ can be chosen to satisfy the identity
$f(x,f(x,y))=f(x,y)$.\\[2mm]
For every $x\in A$, we consider the unary operation
$g_x(y)=f(x,y)$. There is  
a natural number $n_x$ such that $g_x^{n_x}$ is an idempotent transformation
of $A$. Let $n$ be the least common multiple of the $n_x$, $x\in A$ and 
$$
h(x,y)=f(\underbrace{x,f(x,\ldots f(x}_{\mbox{\footnotesize $n$
times}},y)\ldots)).
$$
Since $g^n_x(y)$ is an idempotent for any $x\in A$, we have
$h(x,h(x,y))=g^n_x(g^n_x(y))=g^n_x(y)=h(x,y)$. Finally, as is easily
seen $h$ equals $f$ on $\{a^{\th_{ab}},b^{\th_{ab}}\}$.
\medskip

We prove that, for any $c,d\in A$, the graph $\cG(\Sgn{c,d})$ is 
connected. Moreover, if for every subalgebra $\zB$ of $\zA$,
$\cG(\zB)$ is s-connected, then this holds also for every subalgebra
of $\zA'$. We proceed by induction on order ideals of $\Sub(\zA')$. To
prove the base case for induction, suppose that for $c,d\in\zA'$, the
algebra $\Sgn{c,d}$ is strictly simple. By Theorem~\ref{str}, we
have to consider five cases.
\medskip

\noindent
{\sc Case 1.A.} $\Sgn{c,d}$ is a set.\\[2mm]
In this case, $\Sg{c,d}=\{c,d\}$ and $f\red{\{c,d\}}(x,y)=x$. If 
$\Sgo{c,d}\ne\{c,d\}$ then there exists a term operation $g$ of
$\zA'$ such that $g(c,d)\not\in\{c,d\}$. As is easily seen, the
operation $g'(x,y)=g(f(x,y),f(y,x))$ equals $f$ on
$B'$; hence, it belongs to
$F'$. However, $g'(c,d)=g(c,d)\not\in\{c,d\}$, a contradiction with
the assumption made. Thus, $\Sgo{c,d}=\{c,d\}$. 

Then there is a term operation $g$ of $\zA$ which is either an affine 
or majority or semilattice operation on $\{c,d\}$. The operation
$$
g'(x,y,z)=g(f(x,f(y,z)),f(y,f(z,x)),f(z,f(x,y)))
$$ 
in the first two cases or
$g'(x,y)=g(f(x,y),f(y,x)$ in the latter case belong to $F'$ and is an
affine or majority or semilattice operation on $\{c,d\}$ respectively.
\medskip

\noindent
{\sc Case 1.B.} $\Sgn{c,d}=\{c,d\}$ is a 2-element
semilattice.\\[2mm] 
There is nothing to prove in this case.
\medskip

\noindent
{\sc Case 1.C.} $\Sgn{c,d}$ is a module.\\[2mm] 
The operation $f$ on $\Sgn{c,d}$ has the form $f(x,y)=px+(1-p)y$
and either $p$ or $1-p$ is invertible. Suppose that $p$ is invertible
and $p^n=1$ for a certain $n$. Then set 
$$
f'(x,y)=\underbrace{f(f(\ldots f}_{\mbox{\footnotesize $n$ times}}(x,y) 
\ldots,y),y).
$$
Since $f$ and $f'$ are idempotent, $f'(x,y)=x$ on $\Sgn{c,d}$ and
$f'(x,y)=f(x,y)$ on $B'$. 

Then, as in Case~1.A we show that
$\Sgn{c,d}=\Sgo{c,d}$. Therefore, $\zB=\Sgo{c,d}$ is a
strictly simple algebra. If $\zB$ is 2-element then we get one of the
previous cases. Otherwise, $\zB$ either has a zero-multiplication
operation $h$ or it is of the form $(B;F(H))$ for a certain
permutation group $H$. In the former case, $h(f(x,y),f(y,x))$ belongs
to $F'$ and is a zero-multiplication operation on $\zB$. In the latter 
case, $\zB$ has an operation which is either a semilattice or majority 
operation on $\{c,d\}$. Arguing as above we get an operation of
$\zA'$ which is semilattice or majority on $\{c,d\}$ respectively.
\medskip

\noindent
{\sc Case 1.D.} $\Sgn{c,d}$ has a zero-multiplication
operation $h$.\\[2mm] 
Let $0$ be the zero-element. Then $c,d$ are connected by edges $c0$
and $d0$.
\medskip

\noindent
{\sc Case 1.E.} $\Sgn{c,d}$ is of the form $(B;F(H))$ for a
certain permutation group $H$. (Note that this algebra has the Boolean type.)\\[2mm] 
If there is no automorphism $\vf$ in $H$ such that $\vf(c)=d$ and
$\vf(d)=c$, then, by Proposition~\ref{pro:simple}, $\zA'$ has a term
operation $g$ which is a semilattice operation on $\{c,d\}$. So, let
us suppose that there is an automorphism swapping $c$ and $d$. 

If $\Sgo{c,d}$ has no
operation which is semilattice on $\{c,d\}$ then we are
done. Otherwise, let $g$ be a term operation of $\zA$ semilattice on
$\{c,d\}$ and $h$ a term operation of $\zA'$ majority on
$\{c,d\}$. If one of $h(x,x,y),h(x,y,x),h(y,x,x)$ is a semilattice
operation on $B'$ then we proceed as before. Otherwise, $h$ is a
projection on $B'$; without loss of generality let it be the first
projection. Then $h'(x,y)=f(h(x,y,y),h(y,y,x))$ equals $f$ on 
$B'$ and is a projection on $\{c,d\}$. We complete the proof as
before. 
\medskip

Now, suppose that the claim proved for all proper subalgebras of
$\zC=\Sgn{c,d}$. We consider two cases.
\medskip

\noindent
{\sc Case 1.1.} There is a maximal congruence $\th$ of
$\zC$ such that $f\fac\th$ is commutative on $\zC\fac\th$.\\[2mm]
By Claim 1, $f$ is a semilattice operation on
$\{c^\th,f(c^\th,d^\th)\}$ and
$\{f(c^\th,d^\th),d^\th\}$.
\medskip

\noindent
{\sc Case 1.2.} The set $D=\{(f(c',d'),f(d',c'))\mid c',d'\in\zC\}$ 
generates the total congruence of $\zC$.\\[2mm]
If, for every pair $(g(c,d,c''),g(c,d,d''))$, where $(c'',d'')\in D$, and $g$ is a term operation of $\zC$,
the subalgebra $\Sgo{g(c,d,c''),g(c,d,d'')}$ of $\zA'$ is a
proper subalgebra of 
$\Sgo{c,d}$, then $\cH(\zC)$ is connected and, therefore $c,d$ are connected 
by induction hypothesis. Indeed, if
$\Sgn{g(c,d,c''),g(c,d,d'')}=\zC$, then
$c,d\in\Sgo{g(c,d,c''),g(c,d,d'')}$. Therefore,
$\Sgo{g(c,d,c''),g(c,d,d'')}=\Sgo{c,d}$.

Suppose that, for a certain $(c',d')\in D$ and a ternary term
operation $g$ of $\zC$, we have
$\Sgo{g(c,d,c'),g(c,d,d')}=\Sgo{c,d}$. Then, for any
$e\in\Sgo{c,d}$, there is a term operation $h$ of $\zA$ such that
$h(g(c,d,c'),g(c,d,d'))=e$. Consider
$h'(x,y,z,t)=h(g(x,y,f(z,t)),g(x,y,f(t,z)))$. We have
$h'\red{\{a,b\}}(x,y,z,t)=g(x,y,f(z,t))$, hence, $h'\in F'$. On the
other hand, $h'(c,d,c'',d'')=e$, where $f(c',d')=c''$,
$f(d',c')=d''$. Thus, $\Sgn{c,d}=\Sgo{c,d}$.

The elements $c$ and $d$ are connected by a path $c=e_1,e_2\zd e_k=d$
in $G$. Thus, it is enough to show that if $\Sgo{c,d}$ is
connected by edges of semilattice, majority or affine type in $G$,
then so is $\Sgn{c,d}$. We may assume $cd$ is an edge. Let $\th$ be a 
maximal congruence of
$\Sgo{c,d}$ witnessing that it is an edge and $\th'$ a maximal
congruence of $\Sgn{c,d}$ containing $\th$. 

If $\zB=\Sgn{c,d}\fac{\th'}$ is affine or 2-element we proceed in
the same way as in the base case of induction. If
$\typ(\zB)\in\{\four,\five\}$ then
$c^{\th'},d^{\th'}$ are connected by a chain of 2-element
subalgebras, and the result follows from induction hypothesis.

So, suppose that $\typ(\zB)=\three$. If
$\cH(\zB)$ is connected then we are done by
induction hypothesis. Otherwise we use
Proposition~\ref{pro:simple}. If there is no automorphism $\vf$ of $\zB$
such that $\vf(c^{\th'})=d^{\th'}$ and 
$\vf(d^{\th'})=c^{\th'}$, then, by Proposition~\ref{pro:simple},
$\zA'$ has a term operation $g$ which is a semilattice operation on
$\{c^{\th'},d^{\th'}\}$. So, let us suppose that there is an
automorphism swapping $c^{\th'}$ and $d^{\th'}$.  

If $\Sgo{c,d}$ has no operation which is semilattice on
$\{c^{\th'},d^{\th'}\}$ then we are 
done. Otherwise, let $g$ be a term operation of $\zA$ semilattice on
$\{c^{\th'},d^{\th'}\}$ and $h$ a term operation of $\zA'$
majority on 
$\{c^{\th'},d^{\th'}\}$. If one of $h(x,x,y),h(x,y,x),h(y,x,x)$
is a semilattice operation on $B'$ then we
proceed as before. Otherwise, $h$ is a projection on
$B'$; without loss of
generality let it be the first projection. Then
$h'(x,y)=f(h(x,y,y),h(y,y,x))$ equals $f$ on
$B'$ and is a projection on
$\{c^{\th'},d^{\th'}\}$. We complete the proof as before.
\medskip

Now let $ab$ be of majority type and $m$ a term operation
such that $m\fac{\th_{ab}}$ is a majority operation on
$B'=\{a^{\th_{ab}},b^{\th_{ab}}\}$. Let
$\zA'=(A;F')$ where $F'$ is the set of binary and ternary term
operations $g$ of $\zA$ such that $g\fac{\th_{ab}}$ on $B'$ is either a
projection or equals $m\fac{\th_{ab}}$. As before, the subalgebra of
$\zA$ generated by a set 
$B\sse A$ will be denoted by $\Sgo B$, while the subalgebra of
$\zA'$ generated by the same set will be denote by $\Sgn B$. In
general, $\Sgn B\sse\Sgo B$.
\medskip

\noindent
{\em Claim 2.} $m$ can be chosen to satisfy the identity
$m(x,m(x,y,z),m(x,y,z))=m(x,y,z)$.\\[2mm]
For every $x\in A$, we consider the unary operation $g_x(y)=m(x,y,y)$. There is 
a natural number $n_x$ such that $g_x^{n_x}$ is an idempotent transformation
of $A$. Let $n$ be the least common multiple of the $n_x$, $x\in A$, and 
\begin{eqnarray*}
h(x,y,z)&=&m(\underbrace{x,m(x,\ldots m(x}_{\mbox{\footnotesize $n$
times}},y,z),m(x,y,z)\ldots)),\\
& & m(x,\ldots,y,z),m(x,y,z)\ldots))).
\end{eqnarray*}
Since $g^n_x(y)$ is idempotent for any $x\in A$, we have
$h(x,h(x,y,z),h(x,y,z))=g^n_x(g^n_x(m(x,y,z)))=g^n_x(m(x,y,z))=h(x,y,z)$.
Finally, as is easily seen $h$ is a majority operation on
$B'$.
\medskip

We proceed by induction on order ideals of $\Sub(\zA')$. To prove the
base case for induction, suppose that for $c,d\in\zA'$, the algebra
$\Sgn{c,d}$ is strictly simple. By Theorem~\ref{str}, we have to
consider five cases.
\medskip

\noindent
{\sc Case 2.A.} $\Sgn{c,d}$ is a set.\\[2mm]
In this case, $\Sgn{c,d}=\{c,d\}$ and $m\red{\{c,d\}}(x,y,z)=x$. If 
$\Sgo{c,d}\ne\{c,d\}$ then there exists a term operation $g$ of
$\zA$ such that $g(c,d)\not\in\{c,d\}$. As is easily seen, the
operation $g'(x,y,z)=g(m(x,y,z),m(y,z,x))$ equals $m$ on
$B'$; hence, it belongs to
$F'$. However, $g'(c,d,d)=g(c,d)\not\in\{c,d\}$, a contradiction with
the assumption made. Thus, $\Sgo{c,d}=\{c,d\}$. 

Then there is a term operation $g$ of $\zA$ which is either an affine 
or majority or semilattice operation on $\{c,d\}$. The operation
$$
g'(x,y,z)=g(m(x,y,z),m(y,z,x),m(z,x,y))
$$ 
in the first two cases or
$g'(x,y)=g(m(x,y,z),m(y,z,x)$ in the latter case belong to $F'$ and is an
affine or majority or semilattice operation on $\{c,d\}$ respectively.
\medskip

\noindent
{\sc Case 2.B.} $\Sgn{c,d}=\{c,d\}$ is a 2-element
semilattice.\\[2mm] 
There is nothing to prove in this case.
\medskip

\noindent
{\sc Case 2.C.} $\Sgn{c,d}$ is a module over a ring $K$.\\[2mm] 
The operation $m$ on $\Sgn{c,d}$ has the form
$m(x,y,z)=px+qy+(1-p-q)z$ and either $p$ or $q$ or $1-p-q$ is
invertible. Suppose that $p$ is invertible and $p^n=1$ for a certain
$n$. Then set 
$$
m'(x,y,z,t)=\underbrace{m(m(\ldots m}_{\mbox{\footnotesize $n$
times}}(x,y,y)\ldots,y,y),z,t).
$$
We have, $m'(x,y,z,t)=x+(p-1)y+qz+(1-p-q)t$ on $\Sgn{c,d}$ and
$m'(x,y,z,t)=m(y,z,t)$ on $B'$. Let $k$ be the
characteristics of the ring $K$. We set
$$
m''(x,y,z,t)=\underbrace{m'(m'(\ldots m'}_{\mbox{\footnotesize $k$
times}}(x,y,z,t)\ldots,y,z,t),y,z,t).
$$
For the operation $m''$ we have
$m''(x,y,z,t)=x+k(p-1)y+kqz+k(1-p-q)t=x$ on $\Sgn{c,d}$ and 
$m''(x,y,z,t)=m(y,z,t)$ on $B'$.

Then, as in Case~2.A we show that
$\Sgn{c,d}=\Sgo{c,d}$ (by substituting
$g(m''(x,x,y,z),m''(y,x,y,z))$). Therefore, $\zB=\Sgo{c,d}$ is a 
strictly simple algebra. If $\zB$ is 2-element then we get one of the
previous cases. Otherwise, $\zB$ either has a zero-multiplication
operation $h$ or it is of the form $(B;\cR(H))$ for a certain
permutation group $H$. In the former case,
$h(m''(x,x,y,z),m''(y,x,y,z))$ belongs 
to $F'$ and is a zero-multiplication operation on $\zB$. In the latter 
case, $\zB$ has an operation which is either a semilattice or majority 
operation on $\{c,d\}$. Arguing as above we get an operation of
$\zA'$ which is semilattice or majority on $\{c,d\}$ respectively.
\medskip

\noindent
{\sc Case 2.D.} $\Sgn{c,d}$ has a zero-multiplication
operation $h$.\\[2mm] 
Let $0$ be the zero-element. Then $c,d$ are connected by edges $c0$
and $d0$.
\medskip

\noindent
{\sc Case 2.E.} $\Sgn{c,d}$ is of the form $(B;\cR(H))$ for a
certain permutation group~$H$.\\[2mm] 
If there is no automorphism $\vf$ in $H$ such that $\vf(c)=d$ and
$\vf(d)=c$, then, by Proposition~\ref{pro:simple}, $\zA'$ has a term
operation $g$ which is a semilattice operation on $\{c,d\}$. So, let
us suppose that there is an automorphism swapping $c$ and $d$. 

If $\Sgo{c,d}$ has no
operation which is semilattice on $\{c,d\}$ then we are
done. Otherwise, let $g$ be a term operation of $\zA$ semilattice on
$\{c,d\}$ and $h$ a term operation of $\zA'$ majority on
$\{c,d\}$. If $h$ can be chosen such that it is a majority
operation on $B'$ then we
proceed as before. Otherwise, $h$ is either a projection or minority
or 2/3-minority operation on
$B'$. In the two latter case
one of $h(x,y,y),h(y,y,x),h(y,x,y)$ is the first projection on
$B'$ and the second
projection on $\{c,d\}$; let it be $h(x,y,y)$. 
Then $h'(x,y,z)=m(h(x,y,y),h(y,y,y),h(z,y,y))$ equals $m$ on
$B'$ and is a projection on
$\{c,d\}$. We complete the proof as before.
\medskip

Now, suppose that the claim proved for all proper subalgebras of
$\zC=\Sgn{c,d}$. We consider two cases.
\medskip

\noindent
{\sc Case 2.1.} There is a maximal congruence $\th$ of
$\zC$ such that $m(x,y,y)=m(y,x,y)=m(y,y,x)$ in $\zC\fac\th$.\\[2mm]
We consider the algebra $\zB=\Sgn{c,d}\fac\th$. By the results of
\cite{Kearnes96:idempotent}, $\zB$ is either a set or term equivalent
to a module or has an absorbing element or $\zB^2$ has no skew
congruence. The algebra $\zB$ cannot be a set, because $m$ is not a
projection on $\zB$. If it has an absorbing element, $\cH(\zB)$ is
connected via the absorbing element. In the last case, if $\cH(\zB)$
is connected then the result follows from induction hypothesis. If
$\cH(\zB)$ is disconnected then, by Proposition~\ref{pro:simple},
there is a term operation $f$ of $\zB$ which is either semilattice or
majority on $\{c^\th,d^\th\}$.

Finally, suppose that $\zB$ is term equivalent to module. Then $m$ on
$\zB$ is of the form $m(x,y,z)=px+qy+rz$. Since
$m(x,y,y)=m(y,x,y)=m(y,y,x)$, we have
\begin{eqnarray*}
&& px+(q+r)y=qx+(p+r)y,\\
&& px+(q+r)y=rx+(p+q)y,\\
&& qx+(p+r)y=rx+(p+q)y.
\end{eqnarray*}
Hence,
$$
(p-q)x=(p-q)y,\ (p-r)x=(p-r)y,\ (q-r)x=(q-r)y.
$$
Therefore, $p=q=r$, $m(x,y,z)=px+py+pz$, $3p=1$ and thus $p$ is
invertible. Since $m$ satisfies the identity
$m(x,m(x,y,y),m(x,y,y))=m(x,y,y)$, we have 
\begin{eqnarray*}
m(x,m(x,y,y),m(x,y,y))&=&m(x,y,y)\\
(p+2p^2)x+4p^2y &=& px+2py\\
2p^2x+(4p^2-2p)y &=& 0.
\end{eqnarray*}
This implies $2p^2=(4p^2-2p)=0$ and hence $2p=0$. Comparing this with
the equality $3p=1$ we conclude that the Abelian group of $\zB$ has period 2, and
$m(x,y,z)$ is the minority operation $x+y+z$.

As is easily seen, $h(x,y)=m(x,y,y)$ is the second projection on
$B'$ and the first projection
on $\zB$. Moreover, $h$ can be chosen such that
$h(h(x,y),y)=h(x,y)$. By the induction hypothesis, $c$ and $h(c,d)$
are connected 
by edges (and these edges are of the semilattice and majority types if they are 
such in $\zA$). Let us suppose first that
$\Sgn{d,h(c,d)}\ne\Sgo{d,h(c,d)}$. Then there is a term
operation $g(x,y)$ of $\zA$ such that
$g(h(c,d),d)\not\in\Sgn{d,h(c,d)}$. The 
operation $g(h(x,y),y)=y$ on $B'$ and
$$
g(h(h(c,d),d),d)=g(h(c,d),d)\in\Sgn{d,h(c,d)},
$$
a contradiction.

Since $\th\red{\Sgn{c,d}}$ is a maximal congruence of
$\Sgn{c,d}$, we may assume that
$\Sgn{c,d}=\Sgo{c,d}$. The proof in this case can be
completed in the same way as in Case~2.2.
\medskip

\noindent
{\sc Case 2.2.} The set 
\begin{eqnarray*}
D&=&\{(m(c',d',d'),m(d',c',d')),
(m(c',d',d'),m(d',d',c')),(m(d',d',c'),\\
& & \ \ \ m(d',c',d'))\mid c',d'\in\zC\}
\end{eqnarray*}
generates the total congruence of $\zC$.\\[2mm]
If, for every pair $(g(c,d,c''),g(c,d,d''))$, where $(c'',d'')\in D$ and $g$ is a term 
operation of $\zA'$,
the subalgebra $\Sgo{g(c,d,c''),g(c,d,d'')}$ of $\zA'$ is a
proper subalgebra of $\Sgo{c,d}$, then $\cH(\zC)$ is connected
and, therefore $c,d$ are connected  by induction hypothesis. Indeed, if
$\Sgn{g(c,d,c''),g(c,d,d'')}=\zC$, then
$c,d\in\Sgo{g(c,d,c''),g(c,d,d'')}$. Therefore,
$\Sgo{g(c,d,c''),g(c,d,d'')}=\Sgo{c,d}$.

Suppose that, for a certain $(c',d')\in D$ and a ternary term
operation $g$ of $\zC$, we have
$\Sgo{g(c,d,c'),g(c,d,d')}=\Sgo{c,d}$. Then, for any
$e\in\Sgo{c,d}$, there is a term operation of $\zA$ such that
$h(g(c,d,c'),g(c,d,d'))=e$. Without loss of generality we may assume
that $c'=m(c'',d'',d''), d'=m(d'',c'',d'')$ for certain
$c'',d''\in\zC$. Consider
$h'(x,y,z,t)=h(g(x,y,m(z,t,t)),g(x,y,m(t,z,t)))$. We have
$h'\red{B'}(x,y,z,t)=g(x,y,m(z,t,t))$, hence, $h'\in F'$. On the
other hand, $h'(c,d,c'',d'')=e$. Thus, $\Sgn{c,d}=\Sgo{c,d}$.

The elements $c$ and $d$ are connected by a path $c=e_1,e_2\zd e_k=d$
in $\cG(\Sgo{c,d})$. Thus, it is enough to show that if
$\Sgo{c,d}$ is 
connected by edges of semilattice, majority or affine type in $G$,
then so is $\Sgn{c,d}$. Assume $cd$ is an edge. Let $\th$ be a maximal congruence of
$\Sgo{c,d}$ witnessing that it is an edge and $\th'$ a maximal
congruence of $\Sgn{c,d}$ containing $\th$. 

If $\zB=\Sgn{c,d}\fac{\th'}$ is affine or 2-element we proceed in
the same way as in the base case of induction. If
$\typ(\zB)\in\{\four,\five\}$ then
$c^{\th'},d^{\th'}$ are connected by a chain of 2-element
subalgebras, and the result follows from induction hypothesis.

So, suppose that $\typ(\zB)=\three$. If $\cH(\zB)$ is connected then
we are done by induction hypothesis. Otherwise we use
Proposition~\ref{pro:simple}. 
\end{proof}

\subsection{Unified operations}\label{sec:unified}

To conclude this section we prove that the polymorphisms (or term operations) 
certifying the strict type of edges can be significantly unifying (cf.\ Proposition~2
from \cite{Bulatov03:conservative}).
\begin{theorem}\label{the:uniform}
Let $\zA$ be an idempotent algebra. There are term operations $f,g,h$
of $\zA$ such that 
\begin{description}
\item
$f\red{\{a^{\th_{ab}},b^{\th_{ab}}\}}$ is a semilattice
operation if $ab$ is a strict semilattice edge, it is the first projection if $ab$ is a 
strict majority or affine edge;
\item
$g\red{\{a^{\th_{ab}},b^{\th_{ab}}\}}$ is a majority
operation if $ab$ is a strict majority edge, it is the first
projection if $ab$ is a strict affine edge, and
$g\red{\{a^{\th_{ab}},b^{\th_{ab}}\}}(x,y,z)=
f\red{\{a^{\th_{ab}},b^{\th_{ab}}\}}(x,
f\red{\{a^{\th_{ab}},b^{\th_{ab}}\}}(y,z))$ if $ab$ is 
strict semilattice;
\item
$h\red{\Sg{ab}\fac{\th_{ab}}}$ is an affine operation
operation if $ab$ is a strict affine edge, it is the first
projection if $ab$ is a strict majority edge, and
$h\red{\{a^{\th_{ab}},b^{\th_{ab}}\}}(x,y,z)=
f\red{\{a^{\th_{ab}},b^{\th_{ab}}\}}(x,
f\red{\{a^{\th_{ab}},b^{\th_{ab}}\}}(y,z))$ if $ab$ is strict semilattice.
\end{description}
\end{theorem}

\begin{proof}
Show first that there is an operation $f$ that is semilattice on
each semilattice edge. Let $\vc Bn$ be the list of all semilattice edges in the graph 
$\cG(\zA)$. To avoid clumsy notation we shall denote the operation
$(f\fac{\th_{ab}})\red B_j$, $B_j=\Sg{a,b}$ simply by $f\red{B_j}$. Let also
$\vc fn$ be the list of term operations of the algebra such that
$f_i\red{B_i}$ is a semilattice operation. Notice that every binary
idempotent operation on a 2-element set is either a projection or a
semilattice operation, and every binary operation of a module can be
represented in the form $px+(1-p)y$. Since each $f$ is idempotent,
for any $i,j$, $f_i\red{B_j}$ is either a projection, or a semilattice
operation. We prove by induction, that the operation $f^i$ constructed
via the following rules is a semilattice operation on $\vc Bi$:
\begin{itemize}
\item
$f^1=f_1$;
\item
$f^i(x,y)=f_i(f^{i-1}(x,y),f^{i-1}(y,x))$.
\end{itemize}

The base case of induction, $i=1$ holds by the choice of
$f_1$. Suppose that $f^{i-1}$ satisfies the required conditions. If
$f^{i-1}\red{B_i}$ is a projection, say, $f^{i-1}\red{B_i}(x,y)=x$,
then 
$$
f^i(x,y)=f_i(f^{i-1}(x,y),f^{i-1}(y,x))=f_i(x,y),
$$
that is a semilattice operation on $B_i$. Let $B_i=\{a,b\}$, and $f^{i-1}$ a
semilattice operation such that $f^{i-1}(a,b)=f^{i-1}(b,a)=a$. Then
\begin{eqnarray*}
f^i(a,b) &=& f_i(f^{i-1}(a,b),f^{i-1}(b,a))=f_i(a,a)=a,\\
f^i(b,a) &=& f_i(f^{i-1}(b,a),f^{i-1}(a,b))=f_i(a,a)=a,
\end{eqnarray*}
hence, $f^i$ is again a semilattice operation. 

Thus, for each edge $B$, $f^n\red B$ is a semilattice operation if $B$
is red and either a semilattice operation or a projection or $px+(1-p)y$
otherwise. However, if $B$ is not red, then the subalgebra with the
universe $B$ has no semilattice operation, therefore, $f^n\red B$ is
a projection or $px+(1-p)y$ whenever $B$ is yellow or blue. Arguing
as in the previous section, one can transform $f^n$ such that it
become a projection on blue edges. Finally, 
it is easy to check that $f(x,y)=f^n(f^n(x,y),x)$ satisfies the conditions of the
proposition.

Now let $\vc Bk$, $\vc Cl$ be the lists of all yellow and all blue
    edges respectively, and $\vc gk$, $\vc hl$ the lists of 
    term operations of the algebra such that $g_i\red{B_i}$ is an
    affine operation, and $h_i\red{C_i}$ is the minority
    operation. Notice first, that since neither
$\zB_i=(B_i;F\red{B_i})$ nor $\zC_i=(C_i;F\red{C_i})$ has a 
    term semilattice operation, every their binary term operation is 
    either a projection or, for blue edges an operation of the form $px+(1-p)y$. 
Therefore, for any $i,j$,
    $g_i\red{B_j}(x,y,y), g_i\red{B_j}(y,x,y), g_i\red{B_j}(y,y,x)$,
    $h_i\red{B_j}(x,y,y), h_i\red{B_j}(y,x,y),\lb%
    h_i\red{B_j}(y,y,x)\in\{x,y\}$, and 
    $g_i\red{C_j}(x,y,y), g_i\red{C_j}(y,x,y), g_i\red{C_j}(y,y,x)$,%
    $h_i\red{C_j}(x,y,y)$, $h_i\red{C_j}(y,x,y), h_i\red{C_j}(y,y,x)
    \in\{x,y, px+(1-p)y\}$. This means that the operations $g_i\red{B_j},
    h_i\red{B_j}$ are of one of the
    following types: a projection, the minority operation, the
    majority operation, a 2/3-{\em minority operation}, that is an
    operation satisfying the equalities $m(x,y,y)=y$,
    $m(y,x,y)=m(y,y,x)=x$ or similar. 

First we prove by induction that for every $1\le j\le k$ there is an
operation $g^j(x,y,z)$ which is majority on $B_i$ for $i\le j$. The
operation $g^1=g_1$ gives the base case of induction. Let us assume
that $g^{j-1}$ is already found. If $g^{j-1}\red{B_j}$ is the majority
operation, set $g^j=g^{j-1}$. Otherwise, it is either a projection, or
a 2/3-minority operation, or the minority operation. In all these case
its variables can be permuted such that
$g^{j-1}\red{B_j}(x,y,y)=x$. Then the operation
$p(x,y)=g^{j-1}(x,y,y)$ satisfies the conditions $p\red{B_j}(x,y)=x$,
and $p\red{B_i}(x,y)=y$ for all $i<j$. It is not hard to see that the
operation
$$
g^j(x,y,z)=p(g_j(x,y,z),g^{j-1}(x,y,z))
$$
satisfies the required conditions. 

Further, consider the operation $g^k$. Its restriction $g^k\red{C_j}$,
$1\le j\le l$, is either a projection, or the minority operation. If
$g^k\red{C_j}$ is an operation $px+qy+(1-p-q)z$, then using the
methods of the previous section we can derive
an operation $p(x,y)$ such that $p\red{B_i}(x,y)=y$ for all $1\le i\le
k$, and $p\red{C_j}(x,y)=x$. The operation $g'(x,y,z)=p(x,g^k(x,y,z))$
is majority on $B_i$, $1\le i\le k$, a projection on $C_j$. Therefore,
$g^k\red{C_i}$ can be assumed to be a projection for all $1\le i\le
l$. Then for the operation  
$$
g''(x,y,z)=g^k(x,g^k(y,x,y),g^k(z,z,x))
$$
we have
\begin{eqnarray*}
g''\red{B_i}(x,y,z) &=& g^k\red{B_i}(x,g^k\red{B_i}(y,x,y),
g^k\red{B_i}(z,z,x))\\
&=& g^k\red{B_i}(x,y,z), \quad \hbox{for any $1\le i\le
k$};\\ 
g''\red{C_i}(x,y,z) &=& g^k\red{C_i}(x,g^k\red{C_i}(y,x,y),
g^k\red{C_i}(z,z,x))= x,\\
& & \ \hbox{for any $1\le i\le l$ such that}\\
& & g^k\red{C_i}(x,y,z)=x;\\ 
g''\red{C_i}(x,y,z) &=&
g^k\red{C_i}(x,g^k\red{C_i}(y,x,y),g^k\red{C_i}(z,z,x))\\
&=& g^k\red{C_i}(y,x,y)=x,\quad \hbox{for any $1\le i\le l$}\\
& & \quad \hbox{such that $g^k\red{C_i}(x,y,z)=y$};\\ 
g''\red{C_i}(x,y,z) &=&
g^k\red{C_i}(x,g^k\red{C_i}(y,x,y),g^k\red{C_i}(z,z,x))\\
&=& g^k\red{C_i}(z,z,x)=x,\quad \hbox{for any $1\le i\le l$}\\
& & \quad \hbox{such that $g^k\red{C_i}(x,y,z)=z$}.
\end{eqnarray*}
Finally, to make $g''$ acting correctly on red edges we set 
$$
g(x,y,z)=g''(f(x,f(y,z)),f(y,f(z,x)),f(z,f(x,y))).
$$
The operation $g$ is as required. 

Next we show that for any $1\le j\le l$ there is $h^j$ such that
$h_j\red{C_i}$ is an affine operation for $i\le j$. As usual,
$h^1=h_1$ gives the base case of induction. If $h^{j-1}$ is obtained,
then if $h^{j-1}\red{C_j}$ is an affine operation then set
$h^j=h^{j-1}$. Otherwise, $h^{j-1}\red{C_j}=px+qy+(1-p-q)z$. One of
the coefficients is invertible, let $p$ is invertible and
$p^n=1$. Then set 
$$
h'(x,y,z,t)=\underbrace{h^{j-1}(h^{j-1}(\ldots
h^{j-1}}_{\hbox{\footnotesize $n$ times}}(x,t,t),t,t\ldots t,t),y,z);
$$
we have $h'\red{C_j}=x+qy+(1-p-q)z+(p-1)t$ and $h'\red{C_i}=x-y+z$ for
$i<j$. Furthermore, $h''(x,y,z)=h'(x,y,z,z)$: $h''\red{C_j}=x+qy-qz$,
$h'\red{C_i}=x-y+z$ for $i<j$. If $q$ is invertible then repeating the
procedure above we get $h'''(x,y,z)$ which an affine operation on all
the $C_i$, $i\le j$. Otherwise, $1+q$ is invertible, therefore,
applying the same procedure to $h''(x,x,y)$ we get an operation
$h'''(x,y)$ such that $h'''\red{C_j}=x$ and $h'''\red{C_i}=y$ for
$i<j$. Then to obtain the required operation we set
$h^j=h'''(h_j(x,y,z),h^{j-1}(x,y,z))$. 

Finally, set $p(x,y)=g(x,y,y)$,
$$
\ov h(x,y,z)=p(h^l(x,y,z),x).
$$
and
$$
h(x,y,z)=\ov h(f(x,f(y,z)),f(y,f(z,x)),f(z,f(x,y))).
$$
As is easily seen $h$ satisfies the conditions required.
\end{proof}

\section{Thin edges}\label{sec:thin}

We start with an observation that operations $f,g,h$ identified in 
Theorem~\ref{the:uniform} can be assumed to satisfy certain identities.

\begin{lemma}\label{lem:fgh-identities}
Operations $f,g,h$ found in Theorem~\ref{the:uniform} can be chosen such that
\begin{itemize}
\item[1.]
$f(x,f(x,y))=f(x,y)$ for all $x,y\in\zA$;
\item[2.]
$g(x,g(x,y,y),g(x,y,y))=g(x,y,y)$ for all $x,y\in\zA$;
\item[3.]
$h(h(x,y,y),y,y)=h(x,y,y)$ for all $x,y\in\zA$.
\end{itemize}
\end{lemma}

\begin{proof}
1. Let $f_a(x)=f(a,x)$ for $a\in \zA$. We need to show that $f$ can be chosen such 
that $f_a(f_a(x))=f_a(x)$. Clearly, this can be done by substituting $f(x,f(x,y))$ 
$|\zA|!$ times. It remains to show that every function $f_i(x,y)$ obtained inductively 
from $f_0(x,y)=f(x,y)$ and $f_{i+1}(x,y)=f_i(x,f(x,y))$ is a replacement for $f$. 
That is, for any semilattice edge $ab$, where $\th\in\Con(\Sg{a,b})$ witnesses that $ab$ 
is a semilattice edge,
$$
f_{i+1}(a,b)\eqc\th f_{i+1}(b,a)\eqc\th b.
$$
By induction we have 
\begin{eqnarray*}
f_{i+1}(a,b) &=& f_i(a,f(a,b))\eqc\th f_i(a,b)\eqc\th b,\\
f_{i+1}(b,a) &=& f_i(b,f(b,a))\eqc\th f_i(b,b)\eqc\th b.
\end{eqnarray*}

2. Let $g$ be the operation that is majority on all strict majority edges, and $g_a(x)=g(a,x,x)$. 
We need to show that $g$ can be chosen such that $g_a(g_a(x))=g_a(x)$. Clearly, this 
can be done by substituting $g(x,g(x,y,y),g(x,z,z))$ $|\zA|!$ times. It remains to show 
that every function $g_i(x,y,z)$ obtained inductively from $g_0(x,y,z)=g(x,y,z)$ and 
$g'_{i+1}(x,y,z)=g_i(x,g(x,y,y),g(x,z,z))$ is a replacement for $g$. That is, for any 
strict majority edge $ab$, where $\th\in\Con(\Sg{a,b})$ witnesses that $ab$ is a 
majority edge,
$$
g_{i+1}(a,b,b)\eqc\th g_{i+1}(b,a,b)\eqc\th g_{i+1}(b,b,a)\eqc\th b.
$$
By induction we have 
\begin{eqnarray*}
g_{i+1}(a,b,b) &=& g_i(a,g(a,b,b),g(a,b,b))\eqc\th g_i(a,b,b)\eqc\th b,\\
g_{i+1}(b,a,b) &=& g_i(b,m(b,a,a),g(b,b,b)\eqc\th g_i(b,a,b)\eqc\th b,\\
g_{i+1}(b,b,a) &=& g_i(b,g(b,b,b),g(b,a,a)) \eqc\th g_i(b,b,a)\eqc\th b.
\end{eqnarray*}

3. Let $h_b(x)=h(x,b,b)$ for $b\in\zA$. The goal is to find $m$ such that 
$h_b(h_b(x))=m_b(x)$ for 
all $b$ and all $x$. Clearly, this can be done by substituting $h(h(x,y,y),y,z)$ 
$|\zA|!$ times. It remains to show that every function $h_i(x,y,z)$ obtained 
inductively from $h_0(x,y,z)=h(x,y,z)$ and $h'_{i+1}(x,y,z)=h_i(h(x,y,y),y,z)$ 
is a replacement for $h$. That is, for any affine edge $ab$, 
where $\th\in\Con(\Sg{a,b})$ witnesses that $ab$ is an affine edge,
$$
h_{i+1}(a,b,b)\eqc\th h_{i+1}(b,b,a)\eqc\th a.
$$
By induction we have 
\begin{eqnarray*}
h_{i+1}(a,b,b) &=& h_i(h(a,b,b),b,b)\eqc\th h_i(a,b,b)\eqc\th a,\\
h_{i+1}(b,b,a) &=& h_i(h(b,b,b),b,a) = h_i(b,b,a)\eqc\th a.
\end{eqnarray*}
\end{proof}

\subsection{Semilattice edges}\label{sec:thin-semilattice}

In this section we focus on (strict) semilattice edges of the graph $\cG(\zA)$. Note
first that if one fix a term operation $f$ such that $f$ is a semilattice
operation on every thick semilattice edge of $\cG(\zA)$, then one can define
an orientation of every semilattice edge. A semilattice edge $ab$ is oriented from $a$ 
to $b$ if $f(a^{\th_{ab}},b^{\th_{ab}})=f(b^{\th_{ab}},
a^{\th_{ab}})=b^{\th_{ab}}$. 
Clearly, the orientation strongly depends on 
the choice of the term operation $f$. The graph $\cG(\zA)$ oriented
according to a term operation $f$ will be denoted by $\cG_f(\zA)$. We
then can define {\em semilattice-connected} and {\em strongly semilattice-connected}
components of $\cG_f(\zA)$. We will also use the natural order on the
set of strongly semilattice-connected components of $\cG_f(\zA)$: for
components $A,B$, $A\le B$ if there is a directed path in $\cG_f(\zA)$
consisting of semilattice edges and connecting a vertex from $A$ with a vertex
from $B$. 

We shall now improve the choice of operation $f$ and restrict the kind of
semilattice edges we will use later.
First we show that those semilattice edges $ab$ for which $\th_{ab}$ is not
the equality relation can be thrown out of the graph $\cG(\zA)$ such that
the graph remains connected. Therefore, we can assume that every semilattice
edge $ab$ is such that $f$ is a semilattice operation on $\{a,b\}$. 
 
\begin{prop}\label{pro:thin-thick}
Let $\zA$ be a finite algebra omitting type \one, $f$ a binary term operation 
semilattice on every (thick) semilattice edge and such that $f(x,f(x,y))=f(x,y)$, 
and $\cG'(\zA)$ the subgraph of $\cG(\zA)$ obtained by omitting semilattice 
edges $ab$ such that $\th_{ab}$ is not the equality relation. Then $\cG'(\zA)$ 
is connected. Moreover, if $\cG(\zA)$ is s-connected then $\cG'(\zA)$ 
is semilattice-connected. If $\cG(\zA)$ is sm-connected then $\cG'(\zA)$ is
sm-connected. 
\end{prop}

\begin{proof}
Firstly, by Theorem~\ref{the:adding} we may assume that every thick
semilattice edge of $\zA$ is a subalgebra.
It suffices to show that for any semilattice edge $ab$ (in $\cG_f(\zA)$),
the veritces $a,b$ are connected (s-connected or sm-connected) in $\cG'(\zA)$.
So, assume $\zA=\Sg{a,b}$ and $ab$ is a semilattice edge in $\cG_f(\zA)$.
We proceed by induction on order ideals of the lattice $\Sub(\zA)$ of subalgebras 
of $\zA$. The base case 
of induction, when $\Sg{a,b}$ is strictly simple is obvious,
because $\Sg{a,b}=\{a,b\}$ and there is a semilattice operation on
this algebra.

Let $\th$ be the maximal congruence of $\Sg{a,b}$ witnessing that 
$ab$ is an edge. Let $b'=f(a,b)$, then $f(a,b')=b'\in b^\th$. By the induction
hypothesis $b'$ is connected (s-connected,sm-connected) with $b$ in
$\cG'(b^\th)$. Therefore, we may assume $f(a,b)=b$. If 
$\Sg{a,f(b,a)}\ne\Sg{a,b}$ then we are done, because $b$ and
$f(b,a)$ are connected inside $b^\th$ and $a$ and $f(b,a)$ are
connected inside $\Sg{a,f(b,a)}$. Otherwise there is a term operation 
$t$ such that $t(a,f(b,a))=b$. Then, for the operation
$t'(x,y)=t(x,f(y,x))$, we have
\begin{eqnarray*}
&& t'(a,b)=t(a,f(b,a))=b,\\
&& t'(b,a)=t(b,f(a,b))=t(b,b)=b.
\end{eqnarray*}
Thus, there is a semilattice operation on $\{a,b\}$, hence
$t'\red{\{a,b\}}=f\red{\{a,b\}}$. 
\end{proof}

The graph $\cG'(\zA)$ oriented according to a binary term operation $f$ will be
denoted by $\cG'_f(\zA)$. Semilattice edges $ab$ such that $\th_{ab}$ is
the equality relation will be called \emph{thin semilattice edges}.

Using Proposition~\ref{pro:thin-thick} we are able to impose more
restrictions on the term operation $f$.
\begin{prop}\label{pro:good-operation}
Let $\zA$ be a finite algebra omitting type \one. There is a binary
term operation $f$ of $\zA$ such that $f$ is a semilattice operation on
every thick semilattice edge of $\cG(\zA)$ and, for any $a,b\in\zA$, either
$a=f(a,b)$ or the pair $(a,f(a,b))$ is a semilattice edge of $\cG'_f(\zA)$.
\end{prop}

\begin{proof}
Let $f$ be a binary term operation such that $f$ is semilattice on every 
semilattice edge and $f(x, f(x, y)) = f(x, y)$. Let $a, b\in \zA$ be such that
$f(a, b)\ne a$, and set $b_0 = f(a, b)$ and $b_{i+1} = f(a, f(b_i, a))$
for $i > 0$.

\medskip

{\sc Claim 1.} 
For any $i$, $f(a, b_i) = b_i$.

\smallskip

Indeed, $f(a, b_0) = f(a, f(a, b)) = f(a, b) = b_0$, and for any $i > 0$
$$
f(a, b_i) = f(a, f(a, f(b_{i-1}, a))) = f(a, f(b_{i-1}, a)) = b_i.
$$

Let $\zB_i = \Sg{a, b_i}$. Then $\zB_0\supseteq \zB_1\supseteq\ldots$, and there is $k$
with $\zB_{k+1} = \zB_k$.

\medskip

{\sc Claim 2.}
$f(a, b_k) = f(b_k, a) = b_k$.

\smallskip

Since $b_k\in \zB_{k+1} = \Sg{a, b_{k+1}}$, there is a term operation
$t$ such that $b_k = t(a, b_{k+1})$. Let $s(x, y) =
t(x, f(x, f(y, x)))$. For this operation we have
\begin{eqnarray*}
s(a, b_k) &=& t(a, f(a, f(b_k, a))) = t(a, b_{k+1}) = b_k\\
s(b_k, a) &=& t(b_k, f(b_k, f(a, b_k))) = t(b_k, f(b_k, b_k)) = b_k.
\end{eqnarray*}

This means that $ab$ is a semilattice edge, and the congruence witnessing
it is the equality relation. By the choice of $f$, it is a
semilattice operation on any such pair.

Let $k$ be the maximal among the numbers chosen as before
in Claim 2 for all pairs $a, b$ with $f(a, b)\ne a$. Let $f_0 = f$,
and $f_{i+1}(x, y) = f(x, f(f_i(x, y), x))$ for $i\ge 0$. Let also
$f' = f_k$.

\medskip

{\sc Claim 3.} For any $a, b\in A$, either $f'(a, b) = a$, or the
pair $ac$, where $c = f'(a, b)$ is a semilattice edge witnessed by the
equality relation.

\smallskip

If $f(a, b) = a$ then it is straightforward that $f'(a, b) = a$.
Suppose $f(a, b)\ne a$. We proceed by induction. Since
$\zB_k = \zB_{k+1}$, where the $\zB_i$ are constructed as before,
by Claim 2 $f(a, c) = f(c, a) = c$. This gives the base case
of induction. Suppose $f_i(a, c) = f_i(c, a) = c$. Then
\begin{eqnarray*}
f_{i+1}(a, c) &=& f(a, f(f_i(a, c), a) = f(a, f(c, a)) = c\\
f_{i+1}(c, a) &=& f(c, f(f_i(c, a), c) = f(c, f(c, c)) = c.
\end{eqnarray*}
Claim 3 is proved.

\smallskip

To complete the proof it suffices to check that $f'$ is a
semilattice operation on every (thick) semilattice edge of $\cG(\zA)$.
However, this is straightforward from the construction of
$f'$.
\end{proof}

It will be convenient for us to denote binary operation $f$ that satisfies the 
conditions of Theorem~\ref{the:uniform}, Lemma~\ref{lem:fgh-identities}(1),
and Proposition~\ref{pro:good-operation} by $\cdot$, that is, to write $x\cdot y$
or just $xy$ for $f(x,y)$. The fact that $ab$ is a thin semilattice edge we will
also denote by $a\le b$. In other words, $a\le b$ if and only if $ab=ba=b$.

\begin{lemma}\label{lem:sl-thick-thin}
Let $ab$ be a thick semilattice edge, $\th$ the congruence of $\Sg{a,b}$ that witnesses 
this, and $c\in a^\th$. Then there is $d\in b^\th$ such that $cd$ is a thin semilattice 
edge.
\end{lemma}

\begin{proof}
By Proposition~\ref{pro:good-operation} $cb=c$ or $c\le cb$. Since $d=cb\in b^\th$ 
the former option is impossible. Therefore $cd$ is a thin semilattice edge.
\end{proof}

\subsection{Thin majority edges}\label{sec:thin-majority}

Here we introduce thin majority edges in a way similar to thin semilattice edges,
although in a weaker sense.

\begin{lemma}\label{lem:thin-majority}
Let $\zA$ be an algebra, $ab$ a majority edge in it, and $\th$ the congruence of 
$\Sg{a,b}$ witnessing that. Then there is $b'\in b^\th$ and a ternary term operation 
$g'$ of $\zA$ such that $g'(a,b',b')=g'(b',a,b')=g'(b',b',a)=b'$.
\end{lemma}
\marginpar{Check the `strict' business}

\begin{proof}
Suppose that $b$ is such that $\Sg{a,b}$ is minimal among all subalgebras $\Sg{a,b'}$ 
for $b'\in b^\th$, and such that $g(a,b,b)=b$. Such an element exists by 
Lemma~\ref{lem:fgh-identities}. Consider the ternary relation $\rel$ 
generated by $(a,b,b),(b,a,b),(b,b,a)$. Applying $g$ to these tuples we get 
$(b,b',b'')\in\rel$ for some $b',b''\in b^\th$. Since $b\in\Sg{a,b'}$, say, $t(a,b')=b$, 
$$
\cll bb{b'''}=t\left(\cll bab,\cll b{b'}{b''}\right)\in\rel.
$$
Again, as $b\in\Sg{a,b'''}$, using $(b,b,a),(b,b,b''')\in\rel$ we get $(b,b,b)\in\rel$.
\end{proof}

A majority edge satisying the conditions of Lemma~\ref{lem:thin-majority} will be
called a \emph{thin majority edge}. More precisely, a pair $ab$ is called a thin 
majority edge if (a) it is a majority edge, (b) for any $c\in b^{\th_{ab}}$, 
$b\in\Sg{a,c}$, (c) $g(a,b,b)=b$, and (d) there exists a ternary term operation 
$g'$ such that $g'(a,b,b)=g'(b,a,b)=g'(b,b,a)=b$.  The operation $g$ from 
Theorem~\ref{the:uniform} does not have to satisfy any specific conditions on 
the set $\{a,b\}$, except what follows from its definition. Also, thin majority edges
are directed, since $a,b$ in Lemma~\ref{lem:thin-majority} occur asymmetrically.

\begin{corollary}\label{cor:thin-majority}
For any strict majority edge $ab$, where $\th$ is a witnessing congruence, there is 
$b'\in b^\th$ such that $ab'$ is a thin majority edge.
\end{corollary}

We now consider the interaction of term operations on thin edges in different 
similar algebras.

\begin{lemma}\label{lem:thin-majority-triple}
Let $\zA_1,\zA_2,\zA_3$ be similar idempotent algebras all omitting type \one.
Let $a_1b_1$, $a_2b_2$, and $a_3b_3$ be thin majority edges in $\zA_1,\zA_2,\zA_3$, 
witnessed by congruences $\th_1,\th_2,\th_3$, respectively. Then there is an operation 
$g'$ such that $g'(a_1,b_1,b_1)=b_1$, $g'(b_2,a_2,b_2)=b_2$, $g'(b_3,b_3,a_3)=b_3$.
\end{lemma}

\begin{proof}
Let $\rel$ be the subalgebra of $\zA_1\tm\zA_2\tm\zA_3$ generated by 
$(a_1,b_2,b_3),\lb(b_1,a_2,b_3),(b_1,b_2,a_3)$. Since $a_1b_1$ satisfies condition
(c) of the definition of thin majority edges,
$$
\cll{b_1}{b'_2}{b'_3}=g\left(\cll{a_1}{b_2}{b_3},\cll{b_1}{a_2}{b_3},
\cll{b_1}{b_2}{a_3}\right)\in\rel
$$
and $b'_2\in b_2^{\th_2}, b'_3\in b_3^{\th_3}$. By condition (b)
$b_2\in\Sg{a_2,b'_2}$, in particular, there is a term operation $t$ such that 
$t(a_2,b'_2)=b_2$. Then 
$$
\cll{b_1}{b_2}{b''_3}=t\left(\cll{b_1}{a_2}{b_3},\cll{b_1}{b'_2}{b'_3}\right)\in\rel,
$$
and $b''_3\in b_3^{\th_3}$. Again by condition (b) $b_3\in\Sg{a_3,b''_3}$, in particular, 
there is a term operation $s$ such that $s(a_3,b''_3)=b_3$. Then 
$$
\cll{b_1}{b_2}{b_3}=s\left(\cll{b_1}{a_2}{b_3},\cll{b_1}{b_2}{b''_3}\right)\in\rel.
$$
The result follows.
\end{proof}

\begin{lemma}\label{lem:majority-sl}
Let $\zA_1,\zA_2$ be similar idempotent algebras all omitting type \one.
Let $ab$ be a thin majority edge in $\zA_1$, witnessed by congruences $\th$, and 
$c\le d$ in $\zA_2$. Then there is an binary operation $t$ such that $t(a,b)=b$ and 
$t(d,c)=d$.
\end{lemma}

\begin{proof}
Let $\rel$ be the subalgebra of 
$\zA_1\tm\zA_2$ generated by $(b,c),(a,d)$. Since $ab$ satisfies condition (c) of
the definition of thin majority edges,
$$
\cl bd=g\left(\cl ad,\cl bc,\cl bc\right)\in\rel,
$$
as $g$ is the first projection on semilattice edges. Therefore $t(x,y)=g(x,y,y)$ satisfies 
the conditions.
\end{proof}

\subsection{Thin affine edges}\label{sec:thin-affine}

\begin{lemma}\label{lem:thin-affine}
Let $\zA$ be an algebra, $ab$ an affine edge in it, and $\th$ the congruence of 
$\Sg{a,b}$ witnessing that. Then there is $b'\in b^\th$ and a ternary term operation 
$h'$ of $\zA$ such that  $h'(b',a,a)=h'(a,a,b')=b'$.
\end{lemma}

\begin{proof}
Suppose $a,b$ satisfy the conditions of the lemma. 
Let $b'=h(b,a,a)\in\Sg{a,b}$, by Lemma~\ref{lem:fgh-identities}(3) $h(b',a,a)=b'$. 
As is easily seen, there is $b'$ that satisfies this condition and such that for any 
$b''\in{b'}^{\th'}$ where $\th'$ is the restriction of $\th$ on $\Sg{a,b'}$, it holds
$b'\in\Sg{a,b''}$.

Consider relation $\rel$ generated by pairs $(b',a),(a,a),(a,b')$. Since $h(b',a,a)=b'$, 
$$
\cl{b'}{b''}=h\left(\cl{b'}a,\cl aa,\cl a{b'}\right)\in\rel
$$
and $b''\in {b'}^{\th'}$. By the assumption $b'\in \Sg{a,b''}$, in particular, $(b',b')\in\rel$. 
The result follows.
\end{proof}

Similar to the majority case, an affine edge satisying the conditions of 
Lemma~\ref{lem:thin-affine} will be called a \emph{thin affine edge}. More precisely, 
a pair $ab$ is called a thin majority edge if (a) it is an affine edge, (b) for any 
$c\in b^{\th_{ab}}$, $b\in\Sg{a,c}$, (c) $h(b,a,a)=b$, and (d) there exists a ternary 
term operation $h'$ such that $h'(b,a,a)=h'(a,a,b)=b$.  The operation $h$ from 
Theorem~\ref{the:uniform} does not have to satisfy any specific conditions on 
the set $\{a,b\}$, except what follows from its definition. Also, thin affine edges
are directed, since $a,b$ in Lemma~\ref{lem:thin-affine} occur asymmetrically.

\begin{corollary}\label{cor:thin-affine}
For any affine edge $ab$, where $\th$ is a witnessing congruence, there is $b'\in b^\th$ 
such that $ab'$ is a thin affine edge.
\end{corollary}

\begin{lemma}\label{lem:thin-affine-pair}
Let $\zA_1,\zA_2$ be similar idempotent algebras all omitting type \one.
Let $ab$ and $cd$ be thin affine edges in $\zA_1,\zA_2$, witnessed by congruences 
$\th_1,\th_2$, respectively. Then there is an operation $h'$ such that $h'(b,a,a)=b$ 
and $h'(c,c,d)=d$.
\end{lemma}

\begin{proof}
Let $\rel$ be the subalgebra of $\zA_1\tm\zA_2$ generated by $(b,c),(a,c),(a,d)$. 
By condition (c) of the definition of thin affine edges,
$$
\cl b{d'}=h\left(\cl bc,\cl ac,\cl ad\right)\in\rel
$$
and $d'\in d^{\th_2}$. By condition (b) $\Sg{c,d'}=\Sg{c,d}$, in particular, 
$(b,d)\in\rel$. The result follows.
\end{proof}

\begin{lemma}\label{lem:affine-sl}
Let $\zA_1,\zA_2$ be similar idempotent algebras all omitting type \one.
Let $ab$ be a thin affine edge in $\zA_1$, witnessed by congruences $\th$, and $c\le d$ in $
\zA_2$. Then there is an operation $r$ such that $r(b,a)=b$ and $r(c,d)=d$.
\end{lemma}

\begin{proof}
Let $\rel$ be the subalgebra of $\zA_1\tm\zA_2$ generated by $(b,c),(a,d)$. 
By condition (c) of the definition of thin affine edges,
$$
\cl bd=h\left(\cl bc,\cl ad,\cl ad\right)\in\rel,
$$
as $h(x,y,z)=xyz$ on semilattice edges. The result follows.
\end{proof}

\begin{lemma}\label{lem:affine-maj}
Let $ab$ be a thin affine edge in $\zA_1$, witnessed by congruences $\th$, and $cd$ 
is a thin majority edge in $\zA_2$. Then there is a binary operation $t$ such that 
$t(b,a)=b$ and $t(c,d)=d$.
\end{lemma}

\begin{proof}
Let $\zA_1,\zA_2$ be similar idempotent algebras all omitting type \one.
Let $\rel$ be the subalgebra of $\zA_1\tm\zA_2$ generated by $(b,c),(a,d)$. 
By condition (c) of the definition of thin majority edges,
$$
\cl {b'}d=g\left(\cl bc,\cl ad,\cl ad\right)\in\rel,
$$
where $b'\in b^{\th_1}$, as $g$ is the first projection on $\Sg{a,b}\fac{\th_1}$. 
Then as $b\in\Sg{a,b'}$, we get $(b,d)\in\rel$. The result follows.
\end{proof}

\section{Connectivity}\label{sec:connectivity}

Let $\zA$ be an algebra omitting type \one. A \emph{path} in $\zA$ is a sequence 
$a_0,a_1\zd a_k$ such that $a_{i-1}a_i$ is a thin edge for all $i\in[k]$ (note that thin 
edges are always assumed to be directed). We will 
distinguish paths of several types depending on what types of edges are allowed. 
If $a_{i-1}\le a_i$ for $i\in[k]$ then the path is called a \emph{semilattice} or 
\emph{s-path}. If for every $i\in[k]$ either $a_{i-1}\le a_i$ or $a_{i-1}a_i$ is a thin
affine edge then the path is called \emph{affine-semilattice} or \emph{as-path}. 
Similarly, if only semilattice and thin majority edges are allowed we have a 
\emph{semilattice-majority} or \emph{sm-path}. 
We say that $a$ is \emph{connected} to $b$, $a,b\in\zA$, if there is a path 
$a=a_0,a_1\zd a_k=b$. If this path is semilattice (aftine-semilattice, 
semilattice-majority) then $a$ is said to be \emph{s-connected} (or \emph{as-connected},
or \emph{sm-connected}) to $b$. We denote this by $a\sqq b$ (for s-connectivity),
$a\sqq_{as}b$ and $a\sqq_{sm}b$ for as- and sm-connectivity, respectively.

Let $\cG''_s(\zA)$ denote the digraph whose nodes are the elements of $\zA$, and
the arcs are the thin semilattice edges. The strongly connected component of $\cG_s(\zA)$
containing $a\in\zA$ will be denoted by $\wh a$. The set of strongly connected 
components of $\cG_s(\zA)$ are ordered in the natural way (if $a\le b$ then $\wh a\le \wh b$), 
the elements belonging to maximal ones will be called \emph{maximal}, and
the set of all maximal elements from $\zA$ by $\max(\zA)$. In a similar way
we construct the graph $\cG_{as}(\zA)$ by including all the thin semilattice and 
affine edges. The strongly connected component of $\cG_{as}(\zA)$ containing
$a\in\zA$ will be denoted by $\as(a)$. A maximal strongly connected component 
of this graph is called an \emph{as-component}, an element from an as-component 
is called \emph{as-maximal}, and the set of all as-maximal elements is denoted by 
$\amax(\zA)$.

In this section we show that all maximal elements are connected to each other.
The undirected connectivity easily follows from the definitions, so the challenge 
is to prove directed connectivity, as defined above. We start with an auxiliary lemma.

Let $\rel\le\zA_1\tm\dots\tm\zA_k$ be a relation. Recall that $\tol_i$, $i\in[k]$,  denotes
the \emph{link} tolerance 
\begin{eqnarray*}
&& \{(a_i,a'_i)\in\zA_i^2\mid (a_1\zd a_{i-1},a_i,a_{i+1}\zd a_k),\\
&& \quad 
(a_1\zd a_{i-1},a'_i,a_{i+1}\zd a_k)\in\rel, \text{ for some 
$(a_1\zd a_{i-1},a_{i+1}\zd a_k)$}\}.
\end{eqnarray*}

\begin{lemma}\label{lem:going-maximal}
Let $\zA=\Sg{a,b}$ be simple, $a,b\in\max(\zA)$, and $\rel$ a
subdirect square of $\zA$. Let also $\rela=\tol_2\rel$ be the tolerance defined
by $\{(c,d)\in\zA^2\mid (e,c),(e,d)\in\rel\ \text{ for some } e\}$. If
$\rela$ is a connected tolerance then there is a sequence $a=d_1\zd
d_k=b'$ such that $(d_i,d_{i+1})\in\rela$, $d_i$ is maximal, $b'\in\wh
b$, and if $a_i$ is such that $(a_i,d_i),(a_i,d_{i+1})\in\rel$ then
$a_i$ can also be chosen maximal. 
\end{lemma}

\begin{proof}
We start with any sequence $a=d_1\zd d_k=b$, $(d_i,d_{i+1})\in\rela$
connecting $a$ and $b$. Such a sequence exists because $\rela$ is a
connected tolerance. We prove by induction on $k$.  
The base case of induction is obvious by the choice
of $a$. Suppose $d_i$ is maximal. Let $d_{i+1}=e_1\le\ldots\le e_{s}$ be
a semilattice path and $e_s$ a maximal element. Let also
$(b_j,e_j)\in\rel$ be extensions of the $e_j$ and
$(a_q,d_q),(a_q,d_{q+1})\in\rel$ for $q\in[k-1]$. Then for each $q$, $i\le
q\le k-1$, we construct the sequence $a_q=a_q^1\le\ldots\le a_q^s$ and for
each $q$, $i\le q\le k$, the sequence $d_q=d_q^1\le\ldots\le d_q^s$, where 
$$
a_q^j=a_q^{j-1}\cdot b_j\qquad\mbox{and}\qquad d_q^j=d_q^{j-1}\cdot e_j.
$$
Then observing that
\begin{eqnarray*}
&& \cl{a_q^1}{d_q^1}=\cl{a_q}{d_q}\qquad\mbox{and}\qquad 
\cl{a_q^{j+1}}{d_q^{j+1}}=\cl{a_q^j}{d_q^j}\cdot\cl{b_{j+1}}{e_{j+1}},
\qquad\mbox{and}\\ 
&& \cl{a_q^1}{d_{q+1}^1}=\cl{a_q}{d_{q+1}}\qquad\mbox{and}\qquad 
\cl{a_q^{j+1}}{d_{q+1}^{j+1}}=\cl{a_q^j}{d_{q+1}^j}\cdot
\cl{b_{j+1}}{e_{j+1}}
\end{eqnarray*}
we get that $\cl{a_q^{s}}{d_q^s},\cl{a_q^{s}}{d_{q+1}^s}\in\rel$ for
any $i\le q\le n-1$. Note also that $d_{i+1}^s$ is a maximal element.
Continuing in a similar way we also can guarantee that $a_q^s$ is a
maximal element. 

This process replaces $d_{i+1}$ with a maximal element. However, $d_i$ is also
replaced with another element, and we need to restore the connection of $d_i$ with 
the preceding elements. Since $d_i$ is maximal and $d_i\sqq d_i^s$, these two 
elements belong to the same maximal component. Therefore, there is a semilattice path
$d_i^s=e'_1\le\ldots\le e'_t=d_i$. We now proceed as before. Let $(b'_j,e'_j)\in\rel$ be
extensions of the $e'_j$. Then for each $q$, $i\le q\le k-1$, we construct
sequence $a_q^s=p_q^1\le\ldots\le p_q^t$ and for each $q$, $i\le q\le k$,
sequence $d_q^s=r_q^1\le\ldots\le r_q^t$, where 
$$
p_q^j=p_q^{j-1}\cdot b'_j\qquad\mbox{and}\qquad r_q^j=r_q^{j-1}\cdot e'_j.
$$
Then observing that
\begin{eqnarray*}
&& \cl{p_q^1}{r_q^1}=\cl{a_q^s}{d_q^s}\qquad\mbox{and}\qquad 
\cl{p_q^{j+1}}{r_q^{j+1}}=\cl{p_q^j}{r_q^j}\cdot
\cl{b'_{j-1}}{e'_{j-1}}, \qquad\mbox{and}\\ 
&& \cl{p_q^1}{r_{q+1}^1}=\cl{p_q}{r_{q+1}}\qquad\mbox{and}\qquad 
\cl{p_q^{j+1}}{r_{q+1}^{j+1}}=\cl{p_q^j}{r_{q+1}^j}\cdot\cl{b'_{j-1}}{e'_{j-1}} 
\end{eqnarray*}
we get that $\cl{p_q^t}{r_q^t},\cl{p_q^t}{r_{q+1}^t}\in\rel$ for any
$i\le q\le n-1$. Note also that $r_i^t=d_i$, $r_{i+1}^t$ is a maximal
element, and $r_n^t$ belongs to the same maximal component as $b$.  
\end{proof}

\begin{prop}\label{pro:as-connectivity}
Let $a,b\in\max(\zA)$. Then $a$ is connected to $b$.
\end{prop}

\begin{proof}
We prove the proposition by induction on the size of $\zA$ through a sequence of claims.

\smallskip

{\sc Claim 1.}
$\zA$ can be assumed to be $\Sg{a,b}$.

\smallskip

If $a,b\in\max(\zB)$, $\zB=\Sg{a,b}$, then we are done by the induction hypothesis. 
Suppose they are not and let $c,d\in\max(\zB)$ be such that $a\sqq c$ and $b\sqq d$. 
By the induction hypothesis $c$ is connected to $d$. 
As $a\sqq c$, $a$ is connected to $c$. It remains to show that 
$d$ is connected to $b$. This, however, follows straightforwardly from the assumption 
that $b$ is maximal, and therefore $d\in\wh b$, and so $d\sqq b$ in $\zA$.

\smallskip

{\sc Claim 2.}
$\zA$ can be assumed simple.

\smallskip

Suppose $\zA$  is not simple and $\al$ is its maximal congruence. Let $\zB=\zA\fac\al$. By the 
induction hypothesis $a^\al$ is connected to $b^\al$, that is, there is a sequence $a^\al=
\oa_0,\oa_1\zd \oa_k=b^\al$ such that $\oa_i\le\oa_{i+1}$ or $\oa_i\oa_{i+1}$ is a 
thin affine or majority edge in $\zB$. We will choose some $a_i\in\max(\oa_i)$, where $\oa_i$ is 
viewed as a subalgebra of $\zA$, such that $a_i$ is connected to $a_{i+1}$ in $\zA$. 
Set $a_1=a$.

Depending on whether $\oa_i\oa_{i+1}$ is a semilattice, affine, or majority edge, use 
Lemma~\ref{lem:sl-thick-thin}, \ref{lem:thin-affine}, or~\ref{lem:thin-majority} to
choose $a_{i+1}\in\oa_{i+1}$ such that $a_ia_{i+1}$ is a thin edge.

It remains to show that $a_k$ is connected to $b$. Since $b$ is maximal, it suffices 
to take elements $a',b'$ maximal in $b^\al$ and such that $a_k\sqq a'$ and 
$b\sqq b'$. Then $a_k$ is connected to $b'$ by the induction hypothesis, and $b'$ 
is connected to $b$ in $\zA$, as $b'\in\wh b$.

\smallskip

{\sc Claim 3.}
$\Sg{a,b}$ can be assumed equal to $\Sg{a',b'}$ for any $a'\in\wh a$, $b'\in\wh b$.

\smallskip

If $\Sg{a',b'}\subset \Sg{a,b}$ for some $a'\in\wh a$, $b'\in\wh b$, then by the 
induction hypothesis $a''$ is connected to $b''$ for some $a''\in\wh a$, $b''\in\wh b$. 
Therefore $a$ is also connected to $b$.

\smallskip

Let $\rel$ be the binary relation generated by $(a,b)$ and $(b,a)$. We consider two cases.

\smallskip

{\sc Case 1.}
$\rel$ is not the graph of a mapping, or, in other words, there is no automorphism of $\zA$ that 
maps $a$ to $b$ and $b$ to $a$.

\smallskip

Consider the tolerance $\relo=\tol_1\rel=\{(c,d)\in\zA^2\mid \text{there is $e$ with } (c,e),(d,e)\in\rel\}
$ induced by $\rel$ on $\zA$. Since $\zA$ is simple and $\rel$ is not the graph of a mapping, $\relo
$ is a connected  tolerance. There are again two options.

\smallskip

{\sc Subcase 1a.}
For every $e\in\zA$ the set $B_e=\{d\mid (d,e)\in\rel\}\ne\zA$.

\smallskip

There are $\vc ek\in\zA$ such that $a\in B_{e_1}$, $b\in B_{e_k}$, and 
$B_{e_i}\cap B_{e_{i+1}}\ne\eps$ for every $i\in[k-1]$. By 
Lemma~\ref{lem:going-maximal} the $e_i$'s can be chosen maximal, 
and therefore for every $i\in[k-1]$ we can choose $d_i\in B_i\cap B_{i+1}$ which is 
maximal in $\zA$. For each $i\in[k-1]$ choose $c,d\in\max(B_i)$ with $d_{i-1}\sqq c$ 
and $d_i\sqq d$. By the induction hypothesis $c$ is connected to $d$. Then clearly 
$d_{i-1}$ is connected to $c$, and, as $d_i$ is maximal and $d\in\wh{d_i}$, $d$ is 
connected to $d_i$.

\smallskip

{\sc Subcase 1b.}
There is $e\in\zA$ such that $\zA\tm\{e\}\sse\rel$.

\smallskip

By Lemma~\ref{lem:going-maximal} there are $a'\in\wh a$, $b'\in\wh b$, and a maximal 
element $e'$ such that $(a',e'),(b',e')\in\rel$. Since $\zA=\Sg{a',b'}$, we have 
$\zA\tm\{e'\}\sse\rel$. Thus, $e$ can be assumed maximal. We have therefore 
$(a,b),(a,e),(b,e),(b,a)\in\rel$. If both $\Sg{b,e}$ and $\Sg{e,a}$ are proper 
subalgebras of $\zA$, then proceed as in Subcase 1a. Otherwise suppose 
$\Sg{b,e}=\zA$. This means $\{a\}\tm\zA\sse\rel$, and, in particular, $(a,a)\in\rel$. 
Therefore there is a binary term operation $f$ such that $f(a,b)=f(b,a)=a$, that 
is $b\le a$. Since both elements are maximal, $b\in\wh a$, implying they are connected.

\smallskip

{\sc Case 2.}
$\rel$ is the graph of a mapping, or, in other words, there is an automorphism of 
$\zA$ that maps $a$ to $b$ and $b$ to $a$.

\smallskip

There are two cases to consider.

\smallskip

{\sc Subcase 2a.}
There is no nonmaximal element $c\le a'$ or $c\le b'$ for any $a'\in\wh a$, $b'\in \wh b$. 

\smallskip

If there is a maximal $d$ such that $d'\le d$ for some nonmaximal $d'\in\zA$ 
(that is, $\zA\ne\max(\zA)$), then 
by Case~1 and Subcase~2b $a$ is connected to $d$ and $d$ is connected to $b$.
Suppose all elements in $\zA$ are maximal. 
By Theorem~\ref{the:connectedness} and Proposition~\ref{pro:thin-thick} there are 
$a=a_1,a_2\zd a_k=b$ such that for any $i\in[k-1]$ either $a_ia_{i+1}$ is an 
affine or majority edge (not a thin edge), or $a_i\le a_{i+1}$, or $a_{i+1}\le a_i$. 
In the latter two cases $a_{i+1}\in\wh{a_i}$, and therefore there is a semilattice 
path from $a_i$ to $a_{i+1}$. We need to show that if $a_ia_{i+1}$ is an affine 
or majority edge then $a_i$ is connected 
to $a_{i+1}$. Let $\th$ be a congruence of $\zB=\Sg{a_i,a_{i+1}}$ witnessing that 
$a_ia_{i+1}$ is an affine or majority edge. By Lemmas~\ref{lem:thin-majority} 
and~\ref{lem:thin-affine} there is $b\in a_{i+1}^\th$ such that $a_ib$ is a
thin edge. Then take $c,d\in\max(\zB)$ such that $b\sqq c$ 
and $a_{i+1}\sqq d$. By the induction hypothesis $c$ is connected 
to $d$. Finally, as all elements in $\zA$ are maximal, $d$ is connected with 
$a_{i+1}$ in $\zA$ with a semilattice path.

\smallskip

{\sc Subcase 2b.}
There is a nonmaximal element $c$ with $c\le a'$ or $c\le b'$ for some $a'\in \wh a$ 
or $b'\in\wh b$. In particular, this happens whenever there is a nonmaximal element 
$c$ with $c\sqq a$ or $c\sqq b$.

\smallskip

Note first that we may assume that, for any $b'\in\wh b$, there is an 
automorphism that sends $b'$ to $a$ and $a$ to $b'$, as otherwise we are in the
conditions of Case~1. Recall that we also assume  
$\Sg{a,b'}=\zA$. Because of this and the automorphism swapping $a$ and $b$, 
without loss of generality we may assume that there is nonmaximal $c\le b$. 
Consider $\Sg{a,c}$.

If $\Sg{a,c}=\zA$, consider the relation $\relo$ generated by $(a,c),(c,a)$. Since $c$ 
is not maximal, $\relo$ cannot be the graph of an automorphism. Therefore $\relo$ 
induces a nontrivial tolerance on $\zA$ that, in particular, connects $a$ and $b$, and 
we either complete as in Case~1, or show that $c\le a$, which is a contradiction, as 
$\Sg{a,c}=\{a,c\}$ in this case.

If $\zB=\Sg{a,c}\ne\zA$, take $d\in\max(\zB)$ and such that $c\sqq d$. By the 
induction hypothesis $a$ is connected to $d$. Now let $d\sqq d'$ such that 
$d'\in\max(\zA)$. It remains to show that $d'$ is connected to $b$. If 
$\Sg{d',b}\ne\zA$, the result follows by the induction hypothesis. If there 
is no automorphism that swaps $d'$ and $b$, we argue as in Case~1. So, let 
$\Sg{d',b}=\zA$ and there is an automorphism swapping $d'$ and $b$.

Elements $a,b$ are said to be \emph{v-connected} if there is $c\in\Sg{a,b}$ such 
that $c\sqq a$ and $c\sqq b$. The result follows from the next statement.

\smallskip

{\sc Claim 4.}
If $a,b$ are v-connected and there is an automorphism of $\zA$ that swaps $a$ 
and $b$, then they are connected.

\smallskip

Let $c\sqq d$. The \emph{s-distance} from $c$ to $d$ is the length of the 
shortest semilattice path from $c$ to $d$. The s-distance from $c$ to $\wh d$ is 
the shortest s-distance from $c$ to an element from $\wh d$. The \emph{depth} 
of an element $c$ is the greatest s-distance to a maximal  
component, denoted $\dep(c)$. We prove the Claim by induction on the size of 
$\Sg{a,b}$ and $\dep(c)$, provided $c\sqq a$, $c\sqq b$.

If $c\le a$, $c\le b$, in particular, if $\dep(c)=1$, then there is a binary term 
operation $f$ such that $f(a,b)=c$. Let $d=f(b,a)$. Since there is an automorphism 
swapping $a$ and $b$, $d\le a$ and $d\le b$. 
Set $g(x,y,z)=(f(y,x)\cdot f(y,z))\cdot f(x,z)$. We have
\begin{eqnarray*}
g(a,a,b) &=& (ac)c=a,\\
g(a,b,a) &=& (dd)a=a\\
g(b,a,a) &=& (ca)d=a.
\end{eqnarray*}
Since $a$ and $b$ are automorphic, $g$ is a majority operation on $\{a,b\}$. 
Therefore, $a$ and $b$ are connected.

Suppose the Claim is proved for all algebras and pairs of elements v-connected 
through an element of depth less than $\dep(c)$. Let $c=a_1\le a_2\le a_k=a$ 
and $c=b_1\le b_2\le\ldots\le b_m=b$, and $k>2$ or $m>2$. We may assume 
$b_{m-1}$ is a nonmaximal element and consider $B=\Sg{a,b_{m-1}}$. As in 
Case~1a if $B=\zA$ then there is no automorphism swapping $a$ 
and $b_{m-1}$. Then we consider relation $\relo$ generated by 
$(a,b_{m-1}),(b_{m-1},a)$. We can show that $a$ and $b$ are connected in this case.  

Suppose $B\ne\zA$. Let $d\in\max(B)$ be such that $b_{m-1}\prec d$, let also 
$e\in\max(\zA)$ be such that $d\prec e$. By the induction hypothesis $a$ is connected 
to $d$, and therefore to $e$. Also, $e$ and $b$ are v-connected through $b_{m-1}$, 
and $\dep(b_{m-1})<\dep(c)$. If $\Sg{e,b}\ne\zA$, we conclude by the induction 
hypothesis of the proposition. If $\Sg{e,b}=\zA$, by the induction hypothesis of 
Claim~4, $e$ is connected to $b$.
\end{proof}

\bibliographystyle{plain}

\end{document}